\def\ICNPShepherdVersion{1}
\documentclass[10pt,conference,letterpaper]{IEEEtran}
\IEEEoverridecommandlockouts

\usepackage[T1]{fontenc}
\usepackage{amssymb}
\usepackage{amsmath}
\usepackage{amsthm} 
\usepackage{algorithm}
\usepackage{algorithmicx}
\usepackage{algpseudocode}
\usepackage{comment}
\usepackage{graphicx}
\usepackage{braket}
\usepackage{tikz}
\usepackage[section]{placeins}

\usepackage{bbm}
\usepackage[caption=false,font=footnotesize]{subfig}
\usepackage{textcomp}
\usepackage[colorlinks=true, urlcolor=blue, linkcolor=blue]{hyperref}
\usepackage{cleveref}

\newtheorem{definition}{Definition}

\newtheorem{theorem}{Theorem}

\newtheorem{corollary}{Corollary}
\newtheorem{proposition}{Proposition}

\allowdisplaybreaks

\newcommand{\nop}[1]{}

\newcommand{\ExtendedMaterial}[1]{%
  \ifdefined\ICNPShepherdVersion
    Appendix~\ref{#1}%
  \else
    the extended version%
  \fi
}

\usetikzlibrary{positioning, decorations.pathreplacing}
\usetikzlibrary{backgrounds}

\tikzset{%
  zeroarrow/.style = {-stealth,dashed},
  onearrow/.style = {-stealth,solid},
  c/.style = {circle,draw,solid,minimum width=2em,
        minimum height=2em},
}

\title{Spatio-temporal Path Optimization for Stabilizer-Code-Protected Quantum Networks}

\ifdefined\ICNPShepherdVersion
\else
\IEEEoverridecommandlockouts\IEEEpubid{\makebox[\columnwidth]{979-8-3195-0662-7/26/\$31.00 $\copyright$2026 IEEE \hfill}\hspace{\columnsep}\makebox[\columnwidth]{ }}
\fi

\begin{document}


\author{
\IEEEauthorblockN{
Yuanbo Zhang\IEEEauthorrefmark{1},
Qianfan Wang\IEEEauthorrefmark{2},
Yangming Zhao\IEEEauthorrefmark{3},
Lin Chen\IEEEauthorrefmark{4},
Deke Guo\IEEEauthorrefmark{1}
}

\IEEEauthorblockA{
\IEEEauthorrefmark{1}School of Computer Science and Engineering, Sun Yat-sen University,\\[-1pt]
\texttt{zhangyb36@mail2.sysu.edu.cn}, \texttt{guodk@mail.sysu.edu.cn}\\
\IEEEauthorrefmark{2}Department of Computer Science, City University of Hong Kong, \texttt{qwang742@cityu.edu.hk}\\
\IEEEauthorrefmark{3}State Key Laboratory of Novel Software Technology, Nanjing University, \texttt{ymzhao@nju.edu.cn}\\
\IEEEauthorrefmark{3}School of Intelligent Software and Engineering, Nanjing University Suzhou Campus\\
\IEEEauthorrefmark{4}Engineering Research Centre of Applied Technology on Machine Translation and Artificial Intelligence,\\[-1pt]
Macao Polytechnic University, \texttt{lchen@mpu.edu.mo}
}
\thanks{
The work of Lin Chen was supported by Science and Technology Development Fund of Macau SAR under Grant 0119/2025/RIA2.
The work of Yangming Zhao was supported in part by the National Natural Science Foundation of China under Grants 62572233 and 62272428.
Corresponding authors: Lin Chen and Deke Guo.
Thanks Jifan Liang for assistance with the experimental data.
\ifdefined\ICNPShepherdVersion
This version is an extended version of the ICNP 2026 paper.
\else
An extended version with appendices is available on arXiv at \url{https://arxiv.org/abs/2608.22766}.
\fi
}
}

\maketitle
\begin{abstract} 

Quantum Error Correction~(QEC)-protected direct transmission is a fundamental approach to preserve fragile quantum states while they are physically forwarded across noisy quantum networks. 
When a logical qubit traverses multiple hops, selected QEC-capable nodes may recover the encoded state before it continues along the route. 
The feasibility and cost of the final transmission strategy therefore depend on how we jointly choose the path, the recovery locations, and the protection schemes.
In this paper, we formulate and analyze a cross-layer spatio-temporal path optimization problem for block-style stabilizer-code-protected direct transmission. 
Our main results include fixed-scheme and flexible-scheme single-flow routing algorithms, as well as a multi-flow routing algorithm. 
The framework developed in this paper can serve as an algorithmic building block for QEC-aware routing under logical-error and logical-lifetime constraints. 
Simulations show that it reduces single-flow average routing cost by approximately 25--30\% over Decode-Always and lowers multi-flow throughput-normalized congestion by approximately 28--31\% over Greedy-Assignment.

\end{abstract}


\section{Introduction}

Quantum networks aim to transmit and process quantum information across distributed nodes, enabling applications ranging from secure quantum communication and distributed quantum computing to networked quantum sensing~\cite{zhang2024quantumstack,main2025distributed,kim2024distributed}. 
Unlike classical information, quantum states cannot be copied directly and are vulnerable to channel noise, imperfect operations, and decoherence during multi-hop communication~\cite{kozlowski_towards_2019,cacciapuoti_quantum_2022}. 
Consequently, multi-hop quantum communication is not merely a connectivity problem: routing must specify what quantum object is carried across the network and how its state is protected while traversing noisy channels. 
This paper focuses on the direct-transmission paradigm, where quantum information itself is physically forwarded through noisy channels as encoded logical qubits.\footnote{Another mainstream paradigm is entanglement-based communication, which distributes entangled pairs and uses teleportation to transfer quantum states~\cite{bouwmeester1997experimental,chakraborty2020entanglement,zhao2022e2e,chen_optimum_2024}.}

Direct transmission exposes a concrete protection problem. To protect quantum states, data qubits can be encoded into logical qubits using QEC codes~\cite{nielsen_quantum_2010}. In QEC-protected direct transmission, a logical qubit is carried by physical qubits and forwarded over multiple physical channels. Selected QEC-capable nodes may perform a local recovery operation, namely decoding and re-encoding, to refresh the encoded state before it continues along the route. Such recovery choices, together with the choice of protection scheme, affect both feasibility and resource consumption. Therefore, computing a route for a logical qubit is not only a network-layer path-selection task, but also a joint decision on the physical path, recovery placement, and protection scheme.

%
%
The key challenge is that QEC recovery makes path computation segment-dependent rather than channel-additive.
In classical routing, a path is often evaluated by adding fixed link metrics such as cost or delay.
In QEC-protected direct transmission, however, a channel cannot be evaluated independently because recovery operations define its current segment.
Physical noise and propagation time accumulate until the next recovery operation, and recovery then resets the accumulated segment state at the price of additional QEC operation cost.
\footnote{One may ask how this differs from wireless relaying. Decode-and-forward re-encodes classical messages, while amplify-and-forward forwards noisy signals~\cite{laneman2004cooperative}; QEC recovery refreshes encoded quantum states and sets path segments for logical-error and lifetime checks.}
Thus, path computation must jointly optimize forwarding decisions, recovery placement, and protection-scheme assignment. This joint optimization makes the problem segment-dependent rather than a shortest-path computation over additive channel metrics.

%
%
Two concrete constraints arise.
First, a recovery-delimited segment's logical-error contribution cannot in general be represented by a sum of fixed channel weights. 
Instead, it applies a scheme-dependent nonlinear mapping to the physical noise accumulated since the last recovery. 
As a result, a locally cheap or reliable partial path may become suboptimal after considering where the next recovery operation is performed. 
Second, an encoded logical qubit remains reliable only for a limited lifetime. 
In this paper, the inter-recovery time refers to the elapsed propagation time between two consecutive recovery operations, and each such segment must stay within the logical-lifetime budget of the chosen QEC scheme. 
These two constraints interact directly: performing recovery more frequently shortens each segment and reduces accumulated noise and elapsed time, but increases QEC operation cost; performing recovery less frequently saves cost, but may violate the logical-error or lifetime constraints.

%
%
Current routing models miss this interaction. Classical routing uses additive link metrics and constraints, whereas entanglement-based routing concerns the generation, choice, and purification of entanglement resources~\cite{abane2024entanglement}. Existing studies of QEC-aware direct transmission often treat error correction as a coarse improvement in fidelity or resource allocation~\cite{hu2024qnrouting}. Thus, they model neither segment-dependent, non-additive logical error nor the joint placement of recovery operation, protection schemes, and the maximum travel time before the next recovery.

Motivated by this gap, we make three contributions:
\begin{itemize}
    \item We identify and formulate a QEC-aware spatio-temporal path optimization problem for block-style stabilizer-code-protected direct transmission. 
    Each encoded code block serves as a routing unit. We derive a routing-level model of logical-error evolution from the physical noise accumulated within each segment. The formulation enforces both end-to-end logical-error and inter-recovery logical-lifetime constraints.
%
%
    \item We analyze the computational structure of this problem and show that it is NP-hard even for a single request with a fixed protection scheme. We then develop a generic QEC-aware spatio-temporal state-space framework. It separates QEC-specific segment evaluation from graph-level label-correcting search. Its labels jointly track routing cost, accumulated logical error, and elapsed inter-recovery time under discretization and dominance pruning.
    \item We instantiate the framework in three settings: fixed-scheme routing, flexible protection-scheme assignment, and multi-flow load balancing. In the multi-flow setting, we generate feasible candidate strategies for each request and select among them using an ILP formulation.
\end{itemize}
  
More broadly, this work provides an algorithmic perspective for quantum path optimization under physical-layer protection constraints. 
Rather than treating QEC as a fixed link-quality improvement or a post-processing step, the proposed framework incorporates recovery operations, code-dependent logical-error behavior, and logical lifetime directly into path computation. 
This perspective can support more general QEC-aware routing designs for direct-transmission quantum networks, including networks with heterogeneous protection schemes, resource-aware recovery placement, and network-wide traffic engineering objectives.

\section{Preliminaries}
\label{sec:background}

In this section, we review the essential background of QEC that forms the basis of our network model. 

\subsection{Quantum Error Correction with Stabilizer Codes}

Quantum Error Correction (QEC) combats physical noise and decoherence that threaten quantum information. It protects logical information by encoding it in redundant quantum degrees of freedom~\cite{nielsen_quantum_2010}. 
Unlike classical bits, qubits cannot be copied to create redundancy due to the no-cloning theorem. 
Instead, a QEC code encodes the state of one or more logical qubits into a block of $n$ physical qubits, so that local physical errors can be detected and corrected without directly measuring the encoded logical state.

We use block-style stabilizer-code-based schemes as the main QEC abstraction because a block of physical qubits can be treated as one routing unit, matching our segment-based decode-and-re-encode model. The framework only requires each scheme to expose a few routing-relevant quantities, including its logical-error behavior, decoding/re-encoding cost, and logical decoherence budget; these quantities will be formalized as a scheme profile in Section~\ref{sec:model}. Non-block-style schemes with inter-block or non-local dependencies require an extended state model and are outside the present scope~\cite{grassl2007constructions, wilde2010entanglement}. A stabilizer code is specified by commuting Pauli stabilizer generators~\cite{gottesman1997stabilizer}, whose measurements produce an error syndrome used for decoding and recovery.

The error-correction capability of a stabilizer code is commonly characterized by its parameters $[[n,k,d_{\mathrm{code}}]]$:
\begin{itemize}
    \item \textbf{$n$}: the number of physical qubits in the code block;
    \item \textbf{$k$}: the number of logical qubits protected by the code. In our direct-transmission setting, we focus on $k=1$, i.e., each request carries one logical qubit. The extension to $k>1$ is conceptually straightforward: the $k$ logical qubits encoded in the same block can be treated as one routing unit, with the QoS and cost metrics defined at the block level.
    \item \textbf{$d_{\mathrm{code}}$}: the code distance. Under the standard adversarial error model, a code with distance $d_{\mathrm{code}}$ can correct up to $\lfloor (d_{\mathrm{code}}-1)/2 \rfloor$ arbitrary physical errors.
\end{itemize}

Stabilizer-code families offer different levels of protection capability, physical-qubit overhead, decoding complexity, and operational cost. Examples include surface codes~\cite{fowler_surface_2012} and Steane codes. For routing, we represent these trade-offs through a profile for each candidate protection scheme. This profile specifies the decoding/re-encoding cost. It also records logical-error behavior and logical decoherence time. The routing algorithm can thus compare schemes without modeling the internal details of each code family.

\subsection{Logical Error Rate and Decoherence Time}

The quality of a QEC-protected channel is characterized by logical error rate and logical decoherence time derived from the underlying code and physical noise model.

\subsubsection{Logical Error Rate}

For a given stabilizer code scheme $\sigma$, let $n$ denote the number of physical qubits used to encode one logical qubit. Let $p^{\text{phys}}$ denote the total physical error probability of a physical qubit, and let $p_\text{X}$, $p_\text{Y}$, and $p_\text{Z}$ denote the probabilities of Pauli X~(bit-flip), Z~(phase-flip), and Y~(both) errors, respectively, where $p^{\text{phys}}=p_\text{X}+p_\text{Y}+p_\text{Z}$. 

The logical error rate, $p^{\text{log}}$, is the probability that an uncorrectable error occurs on the logical qubit after a full QEC cycle (i.e., syndrome measurement, decoding, and correction). For a given physical error rate $p^{\text{phys}}$, the logical error rate is well-approximated by~\cite{forlivesi_logical_2024}:
\begin{equation}
\label{eq:logical-error-rate-of-composite-channel}
\begin{aligned}
p^\text{log}_\sigma(p^{\text{phys}}) = & \sum_{j = 0}^{n} \binom {n}{j}(1-p^{\text{phys}})^{n-j} \\ 
        & \cdot \sum _{i = 0}^{j}\binom {j}{i} \, p_{\text {Z}}^{i} \, \sum _{\ell = 0}^{j-i} \binom {j-i}{\ell }\, p_{\text {X}}^{\ell }\, p_{\text {Y}}^{j-i-\ell } f_{\sigma,j}(i,\ell),
\end{aligned}   
\end{equation}
where $j$ denotes the total number of erroneous physical qubits in the code block. Among these $j$ errors, $i$ are $Z$ errors, $\ell$ are $X$ errors, and the remaining $j-i-\ell$ are $Y$ errors. The coefficient $f_{\sigma,j}(i,\ell)$ is a code- and decoder-specific parameter, representing the fraction of uncorrectable error patterns with $i$ $Z$ errors, $\ell$ $X$ errors, and $j-i-\ell$ $Y$ errors under scheme $\sigma$.

\subsubsection{Logical Decoherence Time}
In a multi-hop quantum network, the state of a physical qubit is susceptible to errors induced by decoherence~\cite{schlosshauer_quantum_2019}. These errors are typically modeled by the Pauli operators X, Z, and Y. The rate at which these physical errors occur is characterized by the physical decoherence time, as well as the logical decoherence time. If the time between two decoding operations exceeds the logical decoherence time, the logical qubits suffer an uncorrectable logical error, which cannot be corrected via QEC codes~\cite{schlosshauer_quantum_2019}. To protect logical qubits, it is necessary to decode logical qubits within logical decoherence time~\cite{ryan-anderson_realization_2021}.


\section{System Model and Problem Formulation}
\label{sec:model}

In this section, we develop a novel spatio-temporal network model for QEC-protected quantum networks. 

\subsection{Network Model}

We model a quantum network as an undirected graph $G=(V,E)$.\footnote{We focus on undirected graphs mainly to keep the analysis concise; our results can be extended to directed graphs with straightforward adaptation.} Here, $V$ is the set of quantum nodes, and $E$ is the set of physical channels. 
The network transmits encoded qubits directly through physical channels. It does not first distribute entangled pairs and then teleport quantum states. 
Thus, while encoded logical qubits may contain intra-block quantum correlations, entanglement distribution is not used as a network-level communication primitive in our model.

\subsubsection{Node Model}

The network consists of three types of nodes, distinguished by their QEC capabilities:
\begin{itemize}
    \item \textbf{Ordinary Switches}: These nodes can only forward encoded qubits to the next hop but perform no QEC operations. Hereafter, we refer to them as switches;
    \item \textbf{Super Switches}: These are QEC-capable nodes, at which the routing strategy may perform a decoding-and-re-encoding operation. A decoding node $v \in V$ can perform a full decoding-and-re-encoding operation on a logical qubit;
    \item \textbf{User}: These nodes are request endpoints and are equipped with super switches. We treat user nodes as super switches in our analysis.
\end{itemize}
This heterogeneous model of switches, super switches and users is motivated by the significant resource imbalance between quantum forwarding and quantum error correction. 
Forwarding a logical qubit can be a relatively lightweight operation, while performing a full QEC cycle is an extremely demanding task. 
However, to obtain higher QoS, users are highly likely to be equipped with super switches. Given this consideration, we model the network as heterogeneous.

\subsubsection{Channel Model} Each physical channel $e=(u,v) \in E$ is characterized by three parameters:
\begin{itemize}
    \item An additive traversal cost $c_e$ over channel $e$, which can be instantiated as monetary cost or other resource consumption.
    \item A channel physical error probability $p_e$, which abstracts the probability that a single physical qubit experiences a non-identity physical error during one traversal of channel $e$. For Pauli noise, $p_e=p_{e,\text{X}}+p_{e,\text{Y}}+p_{e,\text{Z}}$, where $p_{e,\text{X}}$, $p_{e,\text{Y}}$, and $p_{e,\text{Z}}$ are the probabilities of X, Y, and Z errors over channel $e$, respectively.
    \item A propagation time $t_e$, the physical time-of-flight for a qubit traversing the channel.
\end{itemize}


\subsection{Strategy and State Evolution}

\subsubsection{Routing and protection strategy formulation}
A routing strategy for a given source-destination pair $(s,d)$ is formally defined as a triplet $\pi=(P,D,\Sigma)$, where:
\begin{itemize}
%
%
    \item $P=(v_0,e_1,v_1,\ldots,e_L,v_L)$ is a simple physical path in the graph $G$, where $v_0=s$, $v_L=d$, $L$ is the path length, $V(P)=\{v_0,\ldots,v_L\}$ is its node set, and $E(P)=\{e_1,\ldots,e_L\}$ is its channel set;
    \item $D \subseteq \{v \in V(P)\mid v \text{ is a super switch}\}$ is the set of chosen decoding locations along the path $P$. Nodes in $V(P)\setminus D$ are designated as forwarding-only;
    \item $\Sigma=\{\sigma_v \mid v\in D\}$ is the protection-scheme allocation, where $\sigma_v \in S_v$ denotes the protection scheme instance assigned for the (re-)encoding operation at node $v$, and $S_v$ is the set of candidate schemes available at node $v$.
\end{itemize}

For routing purposes, each candidate scheme $\sigma \in S_v$ is characterized by a routing-relevant scheme profile
%
%
\[
\phi_\sigma = (c_\sigma,\, T_\sigma^{\log},\, \mathcal F_\sigma),
\]
where $c_\sigma$ is the decoding/re-encoding cost, $T_\sigma^{\log}$ is the logical decoherence budget, and $\mathcal F_\sigma$ maps the accumulated physical noise on a segment to its logical-error probability.
Its Pauli-noise instantiation and the segment-level simplification
used in our evaluation are given by
Eqs.~\eqref{eq:p_comp}--\eqref{eq:simple-logical-error-rate-of-composite-channel}.
The routing algorithm does not directly model the internal operation of the code or decoder. Their effect is represented by $\mathcal F_\sigma$, so the abstraction is only as accurate as this mapping.

%
%
The profile above uses an idealized recovery model. It accounts for the operational cost of decoding and re-encoding but omits recovery faults and latency. In practice, gates, measurements, and decoding may introduce errors, while the recovery process itself takes time. The implementation also determines whether the logical state remains protected during recovery. 
These effects do not alter the structure of our routing framework. 
They can be incorporated by augmenting $\phi_\sigma$ with a scheme-specific recovery error and duration, which are then included in the error and time updates in the algorithm proposed later.

This profile-based recovery abstraction also defines the scope of the current model. It covers block-style stabilizer codes that protect one logical qubit and support node-local recovery. Our model does not directly cover entanglement-assisted stabilizer codes, convolutional stabilizer codes, or other stabilizer-based schemes whose operation fundamentally departs from this block-segment abstraction.

In this paper, we instantiate $S_v$ using surface-code-based schemes. However, the model itself is not restricted to surface codes. Any protection scheme that can provide the above routing-relevant profile can be incorporated into the same optimization framework by replacing the corresponding scheme-specific parameters. If $|S_v|=1$ for all $v$, we refer to this scenario as the fixed-scheme case; otherwise, we refer to it as the flexible-scheme case.

The total cost of a given strategy $\pi=(P,D,\Sigma)$ is the sum of link traversal costs and the operational costs of the assigned protection schemes:
\begin{equation}
C(\pi)=\sum_{e\in E(P)} c_e + \sum_{v\in D} c_{\sigma_v}.
\end{equation}

\subsubsection{Segment-Level Error Modeling}

A path segment $q$ is defined as the subsequence of physical channels traversed by a logical qubit between two consecutive decoding/re-encoding operations. Let $E(q)$ denote the set of physical channels contained in segment $q$. We summarize the accumulated segment-level physical noise by the composite physical error probability
\begin{equation}
    p^{\text{comp}}_q = 1 - \prod_{e \in E(q)} (1 - p_e).
    \label{eq:p_comp}
\end{equation}
%
%
Under independent channel errors, this quantity exactly gives the probability that at least one physical error occurs along the segment. 
Only when $S_q=\sum_{e\in E(q)}p_e\ll 1$ do we use it as a first-order surrogate for the net non-identity Pauli-error probability of the composed segment. 
The omitted multi-error cancellations, e.g., $X\cdot X=I$, contribute $O(S_q^2)$, as shown in \ExtendedMaterial{ap:first-order-composite-error}.
The approximation is not justified when $S_q$ is large.

%
%

To obtain a compact logical-error expression, we use an i.i.d. depolarizing-channel instantiation for the segment-level noise model. Specifically, for a composite physical error probability $p^{\text{comp}}_q$, a physical qubit remains unchanged with probability $1-p^{\text{comp}}_q$; conditioned on an error, it suffers an $X$, $Y$, or $Z$ error with equal probability. Equivalently,
\[
p_\text{X}=p_\text{Y}=p_\text{Z}=p^{\text{comp}}_q/3.
\]

%
%
This instantiation yields a simpler expression based on error weight. It is a reduced form of Eq.~\eqref{eq:logical-error-rate-of-composite-channel}. 
Let $n_\sigma$ be the number of physical qubits used by scheme $\sigma$. 
For weight $j$, $\bar f_{\sigma,j}$ denotes the fraction of Pauli error patterns that remain uncorrectable after decoding. 
This fraction covers all error locations and Pauli error types. 
The coefficient $\bar f_{\sigma,j}$ can be obtained through decoder-level simulation, experimental characterization, or numerical lookup tables. 
Then the logical error rate can be written as
\begin{equation}
\label{eq:simple-logical-error-rate-of-composite-channel}
p^\text{log}_{q,\sigma}(p^{\text{comp}}_q)
=
\sum_{j=0}^{n_\sigma}
\bar f_{\sigma,j}
\binom{n_\sigma}{j}
(1-p^{\text{comp}}_q)^{n_\sigma-j}
(p^{\text{comp}}_q)^j.
\end{equation}
%
%
Our simulations instantiate $\mathcal F_\sigma$ using Eq.~\eqref{eq:simple-logical-error-rate-of-composite-channel} for low-error recovery-delimited segments under i.i.d. depolarizing Pauli noise. 
This instantiation depends only on the Pauli error weight and averages over error locations and Pauli types. It therefore treats all error patterns of the same weight identically and cannot capture distribution-dependent logical behavior. Its i.i.d. assumption also excludes spatial and temporal correlations in practical surface-code channels.
Independent non-depolarizing Pauli noise can use Eq.~\eqref{eq:logical-error-rate-of-composite-channel}; loss, erasure, non-Pauli, or correlated noise requires a scheme-specific mapping $p^\text{log}_{q,\sigma}(\cdot)$ calibrated to the channel and decoder. 
Thus, long-distance repeater links dominated by loss or erasure need a different segment model. 
If the new mapping depends only on the accumulated segment noise represented in the current model, only $\mathcal F_\sigma$ needs to be replaced.
Otherwise, the required noise variables can be added to the labels in the generic state-space framework of Section~\ref{sec:single-flow}.

Before traversing segment $q$, suppose that a logical qubit has a cumulative logical error rate $r_{\text{prev}}$. 
let $p^{\text{log}}_{q,\sigma}$ denote the segment-level logical error rate induced by this segment under protection scheme $\sigma$ after the QEC cycle at the segment endpoint. 
Under the segment-level independence assumption, the cumulative logical error rate $r_{\text{new}}$ after traversing the segment and completing the QEC cycle is
\begin{equation}
    r_{\text{new}} = r_{\text{prev}} + (1 - r_{\text{prev}}) \cdot p^{\text{log}}_{q,\sigma}.
    \label{eq:r_new}
\end{equation}
%
%
Equation~\eqref{eq:r_new} counts two disjoint cases: a logical error exists before the segment, or the segment introduces one when there is no prior error. The update assumes independent segment errors and ignores cross-segment logical-Pauli cancellations.
Under $R=-\log(1-r)$, Eq.~\eqref{eq:r_new} becomes $R_{\mathrm{new}}=R_{\mathrm{prev}}-\log(1-p^{\log}_{q,\sigma})$ and is additive across recovery segments.
Let $a_e=-\log(1-p_e)$ and $A_q=\sum_{e\in E(q)}a_e$, so $p_q^{\mathrm{comp}}=1-e^{-A_q}$.
Writing Eq.~\eqref{eq:simple-logical-error-rate-of-composite-channel} as $F_\sigma$, a segment contributes $\Psi_\sigma(A_q)=-\log(1-F_\sigma(1-e^{-A_q}))$ to $R$.
In general, $\Psi_\sigma(A_1+A_2)\neq \Psi_\sigma(A_1)+\Psi_\sigma(A_2)$, so this contribution is not a sum of fixed physical-channel weights.
A detailed analysis in \ExtendedMaterial{ap:exact-logical-composition} bounds this recurrence error using the net segment-error input $p_{q,\sigma}^{\log,\mathrm{net}}$, which accounts for Pauli-error cancellations within the segment.
It also analyzes the upstream error from Eqs.~\eqref{eq:p_comp}--\eqref{eq:simple-logical-error-rate-of-composite-channel} and gives the low-error conditions for the full approximation chain. Here, $r_0$ is the initial post-encoding logical error rate, set to $0$ unless specified.

\subsection{Single-Flow Problem Formulation}
\label{sec:problem-formulation}
The objective is to compute a minimum-cost strategy $(P^*,D^*,\Sigma^*)$ among all feasible strategies. A strategy is feasible if it satisfies the following two constraints:

\begin{enumerate}
    \item \textbf{End-to-End Error Rate Constraint}: The final logical error rate at the destination node, denoted by $r$, must not exceed the application-specific threshold $r_\theta$, i.e., $r \le r_\theta$.
    
    \item \textbf{Logical Decoherence Time Constraint}: For every segment $q$, the segment travel time $t^{\text{seg}}_q$ must not exceed the logical decoherence budget of the protection scheme $\sigma$ used at the segment start:
    \begin{equation}
        t^{\text{seg}}_q < T^{\text{log}}_\sigma.
    \end{equation}
    Additional channel-level decoherence or storage-loss constraints can be incorporated by refining the segment-time budget, but are not explicitly modeled in this paper.
\end{enumerate}
This formulation defines the single-flow Spatio-Temporal Constrained Path Problem, a state-dependent extension of the classical Multi-Constrained Path Problem (MCPP).

%
%
Classical MCPP assigns each edge fixed additive weights for cost and constrained resources ~\cite{wang1995bandwidth,martins1984multicriteria}. 
Our problem can be cast as MCPP by treating each recovery-delimited segment and protection scheme as one edge in an expanded graph and omitting combinations that violate the lifetime bound.
On the physical graph, however, a segment's accumulated noise is transformed into logical error through a scheme-dependent nonlinear mapping.
Consequently, its logical-error contribution cannot generally be expressed as a sum of fixed weights assigned to individual physical channels.
Our framework avoids enumerating all such segments and scheme choices by retaining the segment state in QEC-aware label transitions to be introduced in the next section.

The following sections instantiate the framework for one fixed scheme (Section~\ref{sec:single-flow}), flexible schemes (Section~\ref{sec:multi-encoding-scheme}), and multi-flow candidate generation (Section~\ref{sec:multi-flows}).

\section{Generic State-Space Framework for Single-Flow Computation}
\label{sec:single-flow}
In this section, we first develop the generic state-space framework for the fixed-scheme single-flow setting. 
In this base case, the protection scheme is fixed for all recovery nodes, so the index $\sigma$ denoting the encoding instance will be omitted, and we use $p^{\text{log}}_{q}, T^{\text{log}}$ instead.
To characterize the intrinsic difficulty of our problem, we first study its decision version. The following theorem shows that even deciding whether a feasible strategy exists is computationally intractable.


\begin{theorem}
\label{th:NP-hard}
The corresponding optimization problem of finding a minimum-cost
feasible strategy is NP-hard.
\end{theorem}

%
%
The proof of Theorem~\ref{th:NP-hard}, given in \ExtendedMaterial{ap:proof-of-theorem-1}, sets all physical error rates to zero. 
Thus, provided that zero physical noise yields zero logical error, the NP-hardness result is independent of the first-order approximation. 
The reduction from the Shortest Weight-Constrained Path Problem (SWCPP) therefore locates the source of NP-hardness in constrained path selection rather than in logical-error evolution. This result of worst case leaves open pseudo-polynomial approaches for suitably bounded instances.

\textbf{Framework Overview.}
%
%
The framework consists of three parts. 
First, state expansion maps node-local forwarding and recovery choices into an auxiliary graph $H=(V_H,E_H)$. 
Second, a QEC-aware label-transition procedure extends a partial label across one auxiliary channel by updating traversal cost, segment noise, elapsed time, and, at recovery nodes, the cumulative logical error. 
Third, a dominance-based label-correcting search keeps only Pareto-efficient labels in the dimensions of cost, logical error, and elapsed inter-recovery time. 
The fixed-scheme case below uses $H=G'$, while the flexible-scheme instantiation in Section~\ref{sec:multi-encoding-scheme} uses $H=G''$.

\textbf{Construction of the Auxiliary Graph.}
We construct the auxiliary graph $G'$ from the original graph $G$ via state-space expansion. The key idea is to convert the node-local operation choice, i.e., forwarding or decoding, into a path choice in the auxiliary graph. In this way, routing and decoding-placement decisions can be handled in a unified manner.

For each original node $v$, let $\mathcal{A}(v)$ denote its corresponding auxiliary node set:
\[
\mathcal{A}(v)=
\begin{cases}
    \{v_{\text{fo}}\}, & \text{if } v \text{ is an ordinary switch},\\
    \{v_{\text{fo}},v_{\text{de}}\}, & \text{if } v \text{ is an intermediate super switch},\\
    \{d_{\text{de}}\}, & \text{if } v=d.
\end{cases}
\]
Here, $v_{\text{fo}}$ represents forwarding the logical qubit without decoding, while $v_{\text{de}}$ represents decoding and re-encoding at super switch $v$. For each original channel $(u,v)\in E$, we add an auxiliary channel $(u',v')$ for every $u'\in\mathcal{A}(u)$ and every $v'\in\mathcal{A}(v)$, with the same traversal cost $c_{u,v}$.

An example of the auxiliary graph $G'$ is shown in Fig.~\ref{fig:graph-model}\subref{fig:extend-graph}. The path $(s_{\text{fo}},u_{\text{fo}},v_{\text{de}},d_{\text{de}})$ represents a strategy that forwards the logical qubit at super switch $u$, and performs decoding at $v$ and $d$.

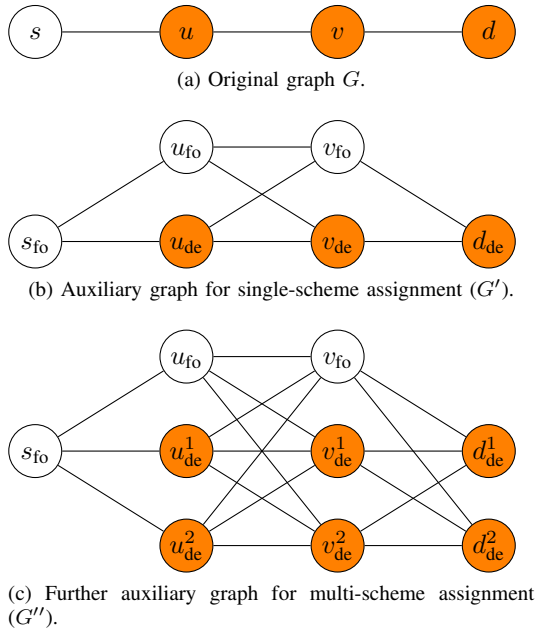
\begin{figure}[tbp]
    \centering

    \tikzset{
        graphTikz/.style={
            auto
        },
        graphNode/.style={
            draw,
            circle,
            minimum size=0.7cm,
            inner sep=0pt,
            text height=1.5ex,
            text depth=.25ex
        },
        superSwitch/.style={
            graphNode,
            fill=orange
        },
        graphEdge/.style={
            draw
        }
    }

    \subfloat[Original graph $G$.]{%
        \begin{tikzpicture}[graphTikz]
            \node[graphNode]   (a) at (0,0) {$s$};
            \node[superSwitch] (b) at (2,0) {$u$};
            \node[superSwitch] (c) at (4,0) {$v$};
            \node[superSwitch] (d) at (6,0) {$d$};

            \begin{scope}[on background layer]
                \draw[graphEdge] (a) -- (b);
                \draw[graphEdge] (b) -- (c);
                \draw[graphEdge] (c) -- (d);
            \end{scope}
        \end{tikzpicture}
        \label{fig:original-graph}
    }

    \subfloat[Auxiliary graph for single-scheme assignment ($G'$).]{%
        \begin{tikzpicture}[graphTikz]
            \node[graphNode]   (a)  at (0,0)    {$s_{\text{fo}}$};
            \node[superSwitch] (b)  at (2,0)    {$u_{\text{de}}$};
            \node[graphNode]   (bp) at (2,1.25) {$u_{\text{fo}}$};
            \node[superSwitch] (c)  at (4,0)    {$v_{\text{de}}$};
            \node[graphNode]   (cp) at (4,1.25) {$v_{\text{fo}}$};
            \node[superSwitch] (d)  at (6,0)    {$d_{\text{de}}$};

            \begin{scope}[on background layer]
                \draw[graphEdge] (a)  -- (b);
                \draw[graphEdge] (a)  -- (bp);
                \draw[graphEdge] (b)  -- (c);
                \draw[graphEdge] (b)  -- (cp);
                \draw[graphEdge] (bp) -- (c);
                \draw[graphEdge] (bp) -- (cp);
                \draw[graphEdge] (c)  -- (d);
                \draw[graphEdge] (cp) -- (d);
            \end{scope}
        \end{tikzpicture}
        \label{fig:extend-graph}
    }

    \subfloat[Further auxiliary graph for multi-scheme assignment ($G''$).]{%
        \begin{tikzpicture}[graphTikz]
            \node[graphNode]   (a)  at (0,0)     {$s_{\text{fo}}$};
            \node[superSwitch] (b1) at (2,0)     {$u_{\text{de}}^1$};
            \node[superSwitch] (b2) at (2,-1.25) {$u_{\text{de}}^2$};
            \node[graphNode]   (bp) at (2,1.25)  {$u_{\text{fo}}$};

            \node[superSwitch] (c1) at (4,0)     {$v_{\text{de}}^1$};
            \node[superSwitch] (c2) at (4,-1.25) {$v_{\text{de}}^2$};
            \node[graphNode]   (cp) at (4,1.25)  {$v_{\text{fo}}$};

            \node[superSwitch] (d1) at (6,0)     {$d_{\text{de}}^1$};
            \node[superSwitch] (d2) at (6,-1.25) {$d_{\text{de}}^2$};

            \begin{scope}[on background layer]
                \draw[graphEdge] (a)  -- (b1);
                \draw[graphEdge] (a)  -- (b2);
                \draw[graphEdge] (a)  -- (bp);

                \draw[graphEdge] (b1) -- (c1);
                \draw[graphEdge] (b1) -- (c2);
                \draw[graphEdge] (b1) -- (cp);

                \draw[graphEdge] (b2) -- (c1);
                \draw[graphEdge] (b2) -- (c2);
                \draw[graphEdge] (b2) -- (cp);

                \draw[graphEdge] (bp) -- (c1);
                \draw[graphEdge] (bp) -- (c2);
                \draw[graphEdge] (bp) -- (cp);

                \draw[graphEdge] (c1) -- (d1);
                \draw[graphEdge] (c2) -- (d1);
                \draw[graphEdge] (cp) -- (d1);

                \draw[graphEdge] (c1) -- (d2);
                \draw[graphEdge] (c2) -- (d2);
                \draw[graphEdge] (cp) -- (d2);
            \end{scope}
        \end{tikzpicture}
        \label{fig:extend-graph-2}
    }

    \caption{An example of transforming the original graph $G$ into the auxiliary graphs $G'$ and $G''$. White nodes represent forwarding-only switches, and orange nodes represent decoding-capable super switches.}
    \label{fig:graph-model}
\end{figure}

\begin{algorithm}[tbp]
\caption{Generic Framework for QEC-Aware Label Correction}
\label{alg:mincost}
\begin{algorithmic}[1]
%
%
    \State \textbf{Input:}auxiliary state-space graph $H=(V_H,E_H)$, source $s$, destination $d$, and discretization steps $\Delta_r,\Delta_t$
    \State \textbf{Output:} a minimum-cost feasible strategy, if one exists

    \For{each node $v\in V_H$}
        \State ${\cal L}_v \leftarrow \emptyset$
    \EndFor
    \State ${\cal L}_s \leftarrow \{(0,\lfloor r_0/\Delta_r\rfloor\Delta_r,0,\emptyset,\emptyset)\}$
    \State $Q \leftarrow \{s\}$

    \While{$Q\neq \emptyset$}
        \State $u \leftarrow \textsc{Dequeue}(Q)$
        \For{each channel $e=(u,v)\in E_H$}
            \For{each label $l\in{\cal L}_u$}
                \State $l_{\text{new}} \leftarrow \textsc{ExtendLabel}(l,e,v)$
                \If{$l_{\text{new}}$ is feasible}
                    \If{\textsc{InsertAndPrune}$(l_{\text{new}},{\cal L}_v)$} 
                        \State \textsc{Enqueue}$(Q,v)$ if $v\notin Q$
                    \EndIf
                \EndIf
            \EndFor
        \EndFor
    \EndWhile

    \If{${\cal L}_d=\emptyset$}
        \State \Return no feasible strategy
    \Else
        \State choose the minimum-cost label in ${\cal L}_d$
        \State reconstruct the path and operations by backtracking predecessors
        \State \Return the corresponding strategy $(\hat P,\hat D,\hat\Sigma)$
    \EndIf
\end{algorithmic}
\end{algorithm}

\begin{algorithm}[tbp]
\caption{\textsc{ExtendLabel}$(l,e,v)$}
\label{alg:extend-label}
\begin{algorithmic}[1]
    \State \textbf{Input:} label $l=(c,\hat r,\hat t,q,\text{pred})$, channel $e=(u,v)$
    \State $q' \leftarrow q\cup\{p_e\}$
    \State $\hat t' \leftarrow \left\lfloor(\hat t+t_e)/\Delta_t\right\rfloor\Delta_t$

    \If{$v$ is a forwarding node or an ordinary switch}
        \State \Return $(c+c_e,\hat r,\hat t',q',u)$
    \EndIf

    \If{$v$ is a decoding node using scheme $\sigma$}
        \If{$\hat t' \ge T^{\text{log}}_\sigma$}
            \State \Return infeasible
        \EndIf
        \State compute $p^{\text{comp}}_{q'}$ by Eq.~\eqref{eq:p_comp}
        \State compute $p^{\text{log}}_{q',\sigma}$ by Eq.~\eqref{eq:simple-logical-error-rate-of-composite-channel}
        \State $r_{\text{new}} \leftarrow \hat r+(1-\hat r)p^{\text{log}}_{q',\sigma}$
        \State $\hat r' \leftarrow \left\lfloor r_{\text{new}}/\Delta_r\right\rfloor\Delta_r$
        \If{$\hat r' > r_\theta$}
            \State \Return infeasible
        \EndIf
        \State \Return $(c+c_e+c_\sigma,\hat r',0,\emptyset,u)$
    \EndIf
\end{algorithmic}
\end{algorithm}

\textbf{Data structure}. For each node $v\in V_H$ in the auxiliary graph, we maintain a label set ${\cal L}_v$. Each label $l\in{\cal L}_v$ represents a partial strategy from the source $s$ to node $v$. Specifically, a label is a quintuple
$l \triangleq (c,\hat r,\hat t,q,v_1)$,
where:
\begin{itemize}
    \item $c$ is the cost accumulated along the partial strategy represented by $l$;
    \item $\hat r$ is the discretized cumulative logical error rate of the logical qubit at node $v$. Given the actual logical error rate $r$, we define
    $\hat r\triangleq \left\lfloor\frac{r}{\Delta_r}\right\rfloor\Delta_r$,
    where $\Delta_r$ is the discretization step size for the error-rate dimension;
    \item $\hat{t}$ is the discretized elapsed time since the last decoding operation. Given the actual elapsed time $t$, we define
    $\hat{t}\triangleq \left\lfloor\frac{t}{\Delta_t}\right\rfloor\Delta_t$,
    where $\Delta_t$ is the time discretization step size;
    \item $q$ is the current segment-noise accumulator. It records the physical error probabilities of the channels traversed since the last decoding operation. It is used to compute the composite physical error probability $p_q^{\text{comp}}$;
    \item $v_1$ is the predecessor of $v$ in the auxiliary graph and is used for path reconstruction.
\end{itemize}
Given two labels $l_i=(c_i,\hat r_i,\hat{t}_i,*,*)$, $i=1,2$, we say that $l_1$ dominates $l_2$ if $c_1\le c_2$, $\hat r_1\le \hat r_2$, and $\hat{t}_1\le \hat{t}_2$. The
dominance is strict if at least one inequality holds strictly. Intuitively, a dominated label cannot lead to a better feasible continuation than the label dominating it. Therefore, dominated labels can be safely removed from ${\cal L}_v$. 

\begin{algorithm}[tbp]
\caption{\textsc{InsertAndPrune}$(l,{\cal L}_v)$}
\label{alg:insert-prune}
\begin{algorithmic}[1]
    \If{$l$ is dominated by some label in ${\cal L}_v$} \label{line:relax}
        \State \Return false
    \EndIf
    \State remove all labels in ${\cal L}_v$ dominated by $l$
    \State add $l$ to ${\cal L}_v$
    \State \Return true
\end{algorithmic}
\end{algorithm}

\textbf{Algorithm}. Algorithm~\ref{alg:mincost} follows a label-correcting search over the auxiliary graph $H$. For each node $v\in V_H$, the algorithm maintains a set of non-dominated labels ${\cal L}_v$, where each label represents a partial strategy from the source $s$ to $v$. Starting from the initial label at $s$, the algorithm repeatedly dequeues a node $u$ and relaxes each outgoing channel $e=(u,v)$ using all labels in ${\cal L}_u$.

%
%
Each channel relaxation invokes the QEC-aware label-transition procedure \textsc{ExtendLabel} in Algorithm~\ref{alg:extend-label}. 
If $v$ is a forwarding node, the label is extended by accumulating the traversal cost, segment time, and model-specific noise state. 
If $v$ is a decoding node, the procedure evaluates $p^{\log}_{q,\sigma}$ from the segment-noise state and updates the cumulative logical error by Eq.~\eqref{eq:r_new}. 
The label is retained only if both the end-to-end error-rate constraint and the logical decoherence-time constraint are satisfied; after decoding, the segment time and segment-noise accumulator are reset.

Algorithm~\ref{alg:insert-prune} uses \textsc{InsertAndPrune} to maintain each label set under the dominance rule. 
It first rejects any feasible label that is already dominated by an existing label at the same auxiliary node. 
Otherwise, it inserts the new label and removes all existing labels dominated by it. 
In this way, each label set stores only Pareto-efficient partial strategies in terms of cost, logical error rate, and elapsed time since the last decoding operation. 
When the search terminates, Algorithm~\ref{alg:mincost} returns the minimum-cost label in ${\cal L}_d$ and reconstructs the corresponding strategy by backtracking predecessors.

%
%
Since the logical error rate and the elapsed segment time are discretized with step sizes $\Delta_r$ and $\Delta_t$, respectively, the number of non-dominated labels stored at each node is bounded by the number of discretized error-time states. Therefore, each label set has size at most $O(\Delta_r^{-1}\Delta_t^{-1})$ within the feasible range. Hereafter, $c(P)$ is the auxiliary-path cost and $r(P)$ is the routing-layer error computed by Eqs.~\eqref{eq:p_comp}--\eqref{eq:r_new}. The guarantee below bounds discretization error relative to this model, not physical-model mismatch.

\begin{definition} 
%
%
    For a segment $q$, let $T^{\text{log}}_q$ be the logical decoherence time of its assigned protection scheme. An $s$-$d$ path $P$ is called an $(\epsilon_r,\epsilon_t)$-optimal feasible path if $c(P) \leq c(P^*)$, $r(P)\le(1+\epsilon_r)r_\theta$, and $t^{\text{seg}}_q(P)\le(1+\epsilon_t)T^{\text{log}}_{q}$ for every segment $q$ of $P$.
\end{definition}

\begin{theorem}
    Let $T^{\text{log}}$ be the logical decoherence time of a path traversing all nodes in $V_H$. If $\Delta_r \leq \epsilon_r \frac{r_\theta}{|V_H|}$ and $\Delta_t \leq \epsilon_t \frac{T^{\text{log}}}{|V_H|}$, where $\Delta_r$ and $\Delta_t$ are asymptotically $O(\epsilon_r/|V_H|)$ and $O(\epsilon_t/|V_H|)$, respectively, then with input $H=(V_H,E_H)$, Algorithm~\ref{alg:mincost} outputs an $(\epsilon_r,\epsilon_t)$-optimal feasible path.
\label{th:mincost}
\end{theorem}

%
%
The proof is in \ExtendedMaterial{ap:proof-mincost}.
The output of our framework allows for relaxed constraints and does not yield an exact cost-approximation ratio.
The same theoretical guarantees apply to a higher-order or exact error model if it uses the existing label state, supports incremental updates, preserves dominance, and has the same asymptotic extension cost. If additional state is required, the label bound and complexity must be re-derived.

\textbf{Space and time complexity of Algorithm~\ref{alg:mincost}.}
Let $T^{\text{log}}_{\max}=\max_{\sigma}T^{\text{log}}_{\sigma}$ and let $N_L$ be the maximum number of labels stored at one node. 
Labels are indexed by discretized logical error rate and elapsed segment time. 
After dominance pruning, $N_L=O\!\left(\frac{r_\theta-r_0}{\Delta_r}\cdot\frac{T^{\text{log}}_{\max}}{\Delta_t}\right)=O(|V_H|^2\epsilon_r^{-1}\epsilon_t^{-1})$.
Here, $\Delta_r=O(\epsilon_r/|V_H|)$ and $\Delta_t=O(\epsilon_t/|V_H|)$. 
Hence, the space complexity of generic framework is $O(|V_H|^3\epsilon_r^{-1}\epsilon_t^{-1})$. The running time is $O(|V_H|^3|E_H|\epsilon_r^{-1}\epsilon_t^{-1})$. 
For the fixed-scheme instantiation $H=G'$, where $|V'|=O(|V|)$ and $|E'|=O(|E|)$, these become $O(|V|^3\epsilon_r^{-1}\epsilon_t^{-1})$ and $O(|V|^3|E|\epsilon_r^{-1}\epsilon_t^{-1})$, respectively.

\section{Flexible-Scheme Single-Flow Instantiation}
\label{sec:multi-encoding-scheme}

We next instantiate the generic framework for the flexible-scheme single-flow variant of the base problem formulated in Section~\ref{sec:problem-formulation}. In this variant, each decoding-capable super switch can choose one protection scheme from its candidate set. Therefore, the algorithm must jointly optimize the path $P$, the decoding placement $D$, and the scheme assignment $\Sigma$. Compared with the fixed-scheme instantiation, the label-correcting core and QEC-aware label-transition procedure remain unchanged; only the state expansion is refined to encode scheme choices.

For clarity of presentation, we first assume that all super switches are equipped with the same set of $b$ candidate protection schemes. The framework naturally extends to the heterogeneous case where different super switches have different candidate sets, by adapting the graph transformation locally for each node.

To solve this problem, we extend the auxiliary-graph construction from $G'$ to a scheme-expanded graph $G''$. For each super switch $v$, we create one forwarding node $v_{\text{fo}}$ and $b$ decoding nodes $v_{\text{de}}^1,\ldots,v_{\text{de}}^b$, where $v_{\text{de}}^i$ represents decoding and re-encoding under protection scheme $\sigma_i$. For each ordinary switch $v$, we create one forwarding node $v_{\text{fo}}$.

For channel reconstruction, let $\mathcal A(v)$ denote the set of auxiliary nodes corresponding to an original node $v$:
\[
\mathcal A(v)=
\begin{cases}
\{v_{\text{fo}}\}, & \text{if } v \text{ is an ordinary switch},\\
\{v_{\text{fo}},v_{\text{de}}^1,\ldots,v_{\text{de}}^b\}, & \text{if } v \text{ is a super switch},\\
\{d_{\text{de}}^1,\ldots,d_{\text{de}}^b\}, & \text{if } v=d.
\end{cases}
\]
For each original channel $(u,v)\in E$, we add an auxiliary channel $(u',v')$ for every $u'\in\mathcal A(u)$ and $v'\in\mathcal A(v)$, with the same traversal cost $c_{u,v}$.

An example of $G''$ is shown in Fig.~\ref{fig:graph-model}\subref{fig:extend-graph-2}, where each super switch has two candidate schemes. For example, the path $(s_{\text{fo}},u_{\text{fo}},v_{\text{de}}^1,d_{\text{de}}^2)$ represents a strategy that forwards the logical qubit at $u$, decodes and re-encodes it at $v$ using scheme $\sigma_1$, and finally decodes it at $d$ using scheme $\sigma_2$.

The joint optimization over routing, decoding placement, and scheme assignment is then obtained by running Algorithm~\ref{alg:mincost} with $H=G''$.

\textbf{Space and time complexity.}
The complexity follows directly from the generic framework with $H=G''$. If each super switch has $b$ candidate schemes, then $|V''|=O(b|V|)$ and $|E''|=O(b^2|E|)$ in the worst case. Therefore, the total space complexity is $O(b^3|V|^3\epsilon_r^{-1}\epsilon_t^{-1})$, and the running time is $O(b^5|V|^3|E|\epsilon_r^{-1}\epsilon_t^{-1})$.

\section{Multi-Flow Network-Wide Instantiation}
\label{sec:multi-flows}

The formulation in Section~\ref{sec:problem-formulation} focuses on a single communication request. 
We now use the generic framework for candidate generation in the multi-flow setting. 
Multiple requests share the same physical network and may compete for channel resources. 
Let $\mathcal{K}$ denote the set of flows, and let $(s_i,d_i)$ be the source-destination pair of flow $i\in\mathcal{K}$. Unlike the single-flow case, where the objective is to minimize the cost of one feasible strategy, the multi-flow setting must also account for network-wide congestion.

We adopt a two-stage approach:
\begin{itemize}
    \item First, we invoke Algorithm~\ref{alg:mincost} with a relaxed dominance rule to generate a pool of candidate feasible strategies for each flow.
    \item Second, we select one candidate strategy for each flow so as to minimize the maximum congestion level over all channels.
\end{itemize}

To generate candidate strategies, we modify the dominance rule in Algorithm~\ref{alg:insert-prune}. Instead of keeping only non-dominated labels, we allow a new label to be inserted if it is dominated by fewer than $M$ existing labels in ${\cal L}_v$. This maintains up to $M$ representative candidates for each discretized state. At the destination, the algorithm returns up to $M$ minimum-cost feasible strategies for each flow; if fewer than $M$ feasible strategies exist, all of them are returned.

Let $\Pi_i$ denote the candidate strategy pool of flow $i$, and let $a_i$ denote its throughput demand. For a strategy $\pi\in\Pi_i$, let $E(\pi)$ denote the set of physical channels used by its path. We formulate the following path-allocation ILP:
\begin{align*}
    \min \quad & y \\
    \text{s.t.}\quad
        & \sum_{\pi\in \Pi_i} x_{i,\pi}=1, && \forall i\in\mathcal{K},\\
        & \sum_{i\in\mathcal{K}}\sum_{\pi\in\Pi_i:\, e\in E(\pi)} a_i x_{i,\pi}\le y, && \forall e\in E,\\
        & x_{i,\pi}\in\{0,1\}, && \forall i\in\mathcal{K},\ \pi\in\Pi_i.
\end{align*}
%
%
The first constraint ensures that each flow selects exactly one candidate strategy. The second constraint upper-bounds the total throughput routed over each physical channel by the maximum congestion variable $y$. The objective therefore minimizes bottleneck congestion.

We solve the LP relaxation and independently select $\pi\in\Pi_i$ for each flow with probability $x^*_{i,\pi}$.
Let $y_{\mathrm{OPT}}^{\Pi}$ denote the optimal integral congestion of the above allocation problem for the fixed candidate pools $\{\Pi_i\}$.

\begin{theorem}
%
%
Assume that throughput demands are normalized so that $0\le a_i\le 1$, and let $n=|V|$. There exists a constant $\kappa>0$ such that, with probability at least $1-1/n$, randomized rounding achieves an $O(\log n/\log\log n)$ approximation to $y_{\mathrm{OPT}}^{\Pi}$.
\end{theorem}

%
%
The theorem guarantees only allocation over the generated candidate pools. 
The standard randomized-rounding proof is given in \ExtendedMaterial{ap:multi-flow-proof}. It does not compare with the global optimum over routing, decoding decisions and so on.
Overall quality therefore depends on whether each $\Pi_i$ contains diverse feasible strategies, particularly alternatives that use different channels. This is affected by $M$ and the relaxed dominance rule.


\section{Performance Evaluation}
\label{sec:simulation}

In this section, we evaluate four aspects of the proposed framework: single-flow routing against static recovery policies, adaptive scheme assignment, multi-flow congestion, and running-time scalability.

\begin{table}[ht]
\centering
\caption{Parameters of Encoding Schemes Used in Simulations}
\label{tab:schemes}
\begin{tabular}{l|ccc}
\hline
\textbf{Scheme Name} & \textbf{[[n,k,$d_{\text{code}}$]]} & \textbf{Cost} & \textbf{Decoherence Time(\textmu s)} \\
\hline
Surface (d=3) & [[13,1,3]] & 5  & 50 \\
Surface (d=5) & [[41,1,5]] & 8  & 100 \\
Surface (d=7) & [[85,1,7]] & 10  & 200 \\
Surface (d=9) & [[145,1,9]] & 50 & 600 \\
\hline
\end{tabular}
\end{table}

\begin{figure*}[!t]
    \centering

    \def\subfigwidththree{0.30\textwidth}

    \subfloat[Core Algorithm Validation: Accept Ratio\label{fig:exp1_acceptance}]{%
        \includegraphics[width=\subfigwidththree]{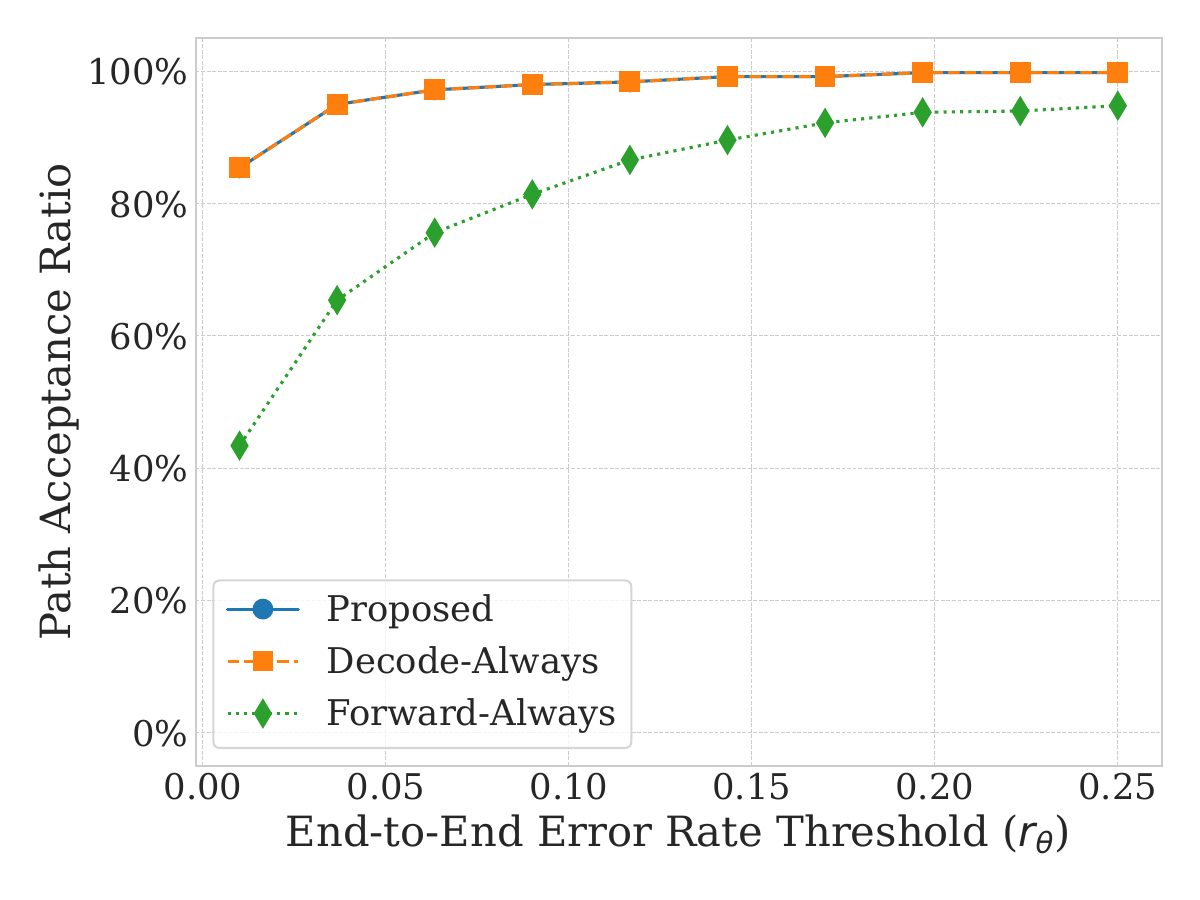}
    }
    \hfill
    \subfloat[Adaptive Scheme Assignment: Accept Ratio\label{fig:exp2_acceptance}]{%
        \includegraphics[width=\subfigwidththree]{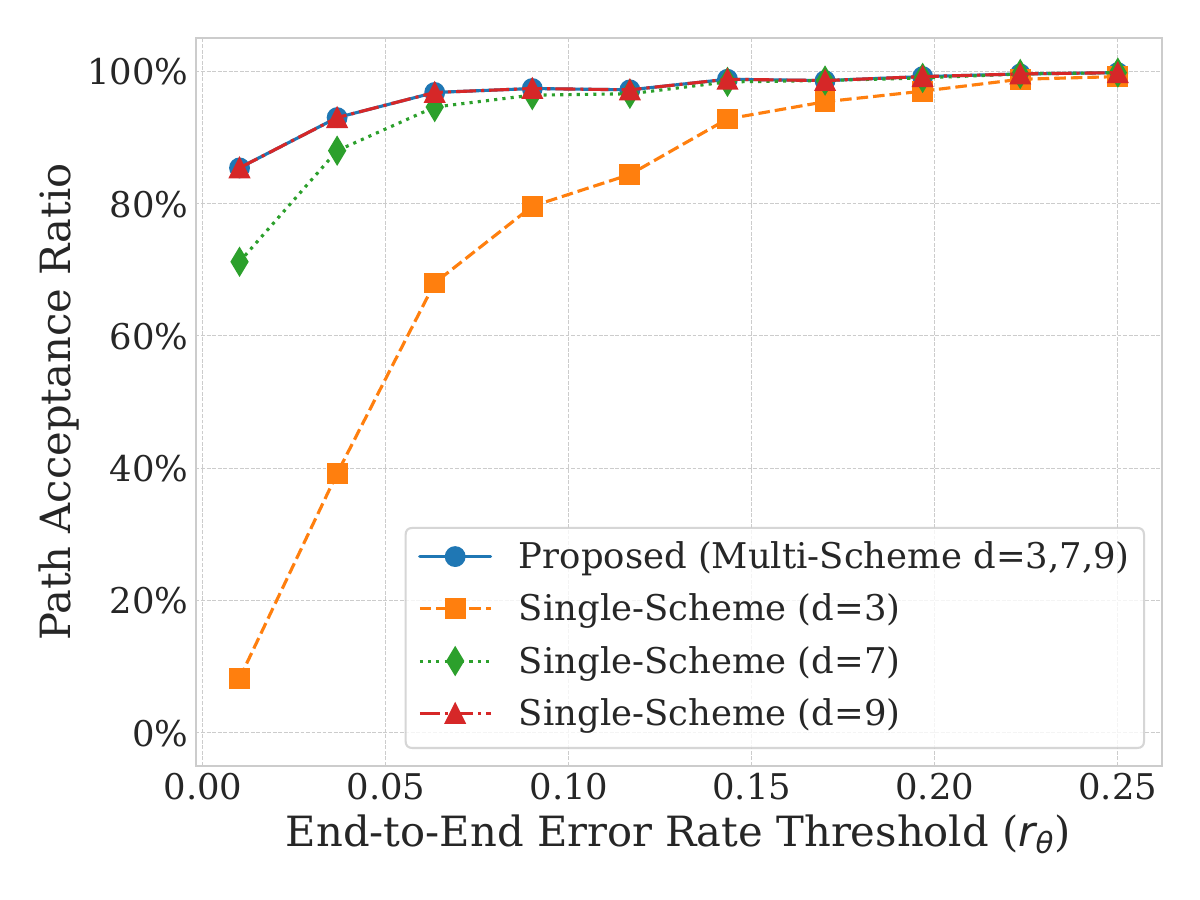}
    }
    \hfill
    \subfloat[Multi-Flow: Maximum Link Congestion\label{fig:exp3_congestion}]{%
        \includegraphics[width=\subfigwidththree]{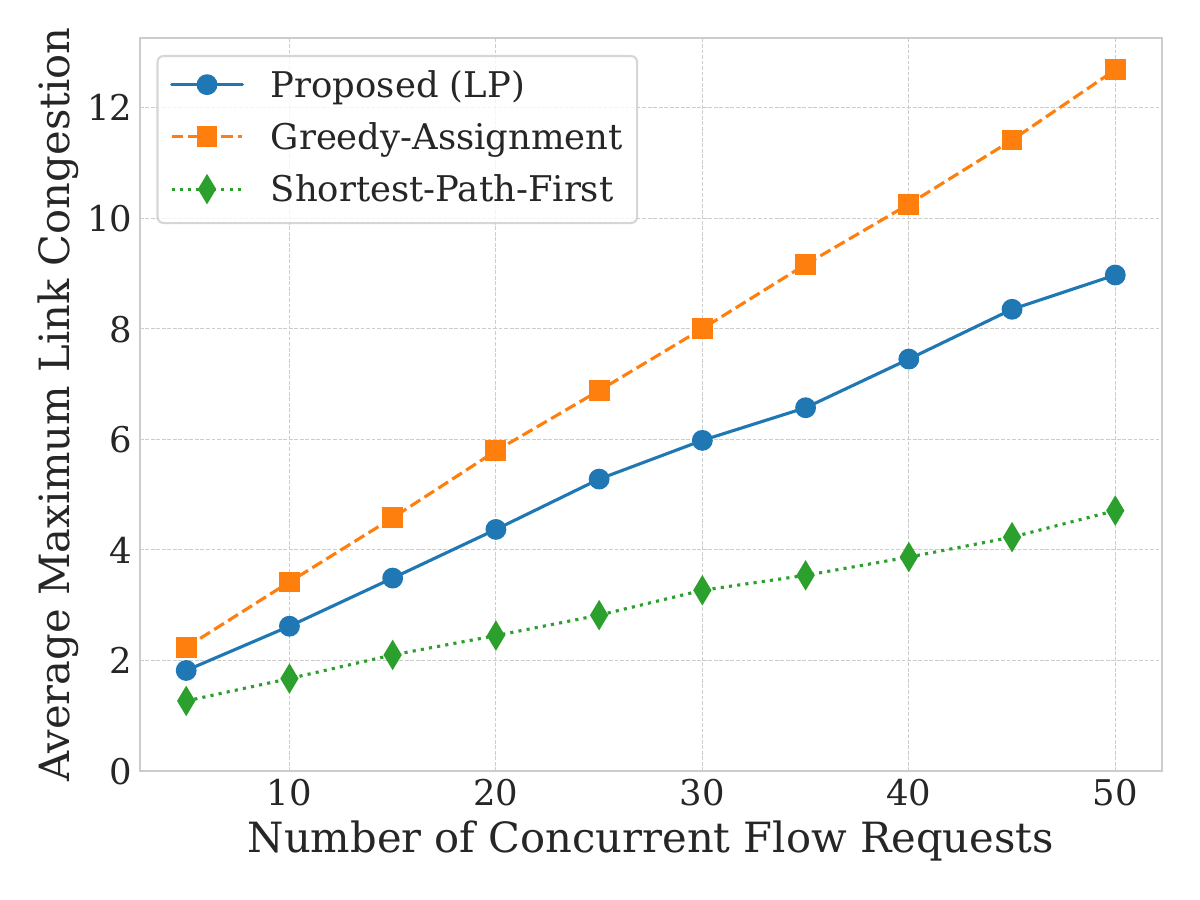}
    }


    \subfloat[Core Algorithm Validation: Avg Cost\label{fig:exp1_cost}]{%
        \includegraphics[width=\subfigwidththree]{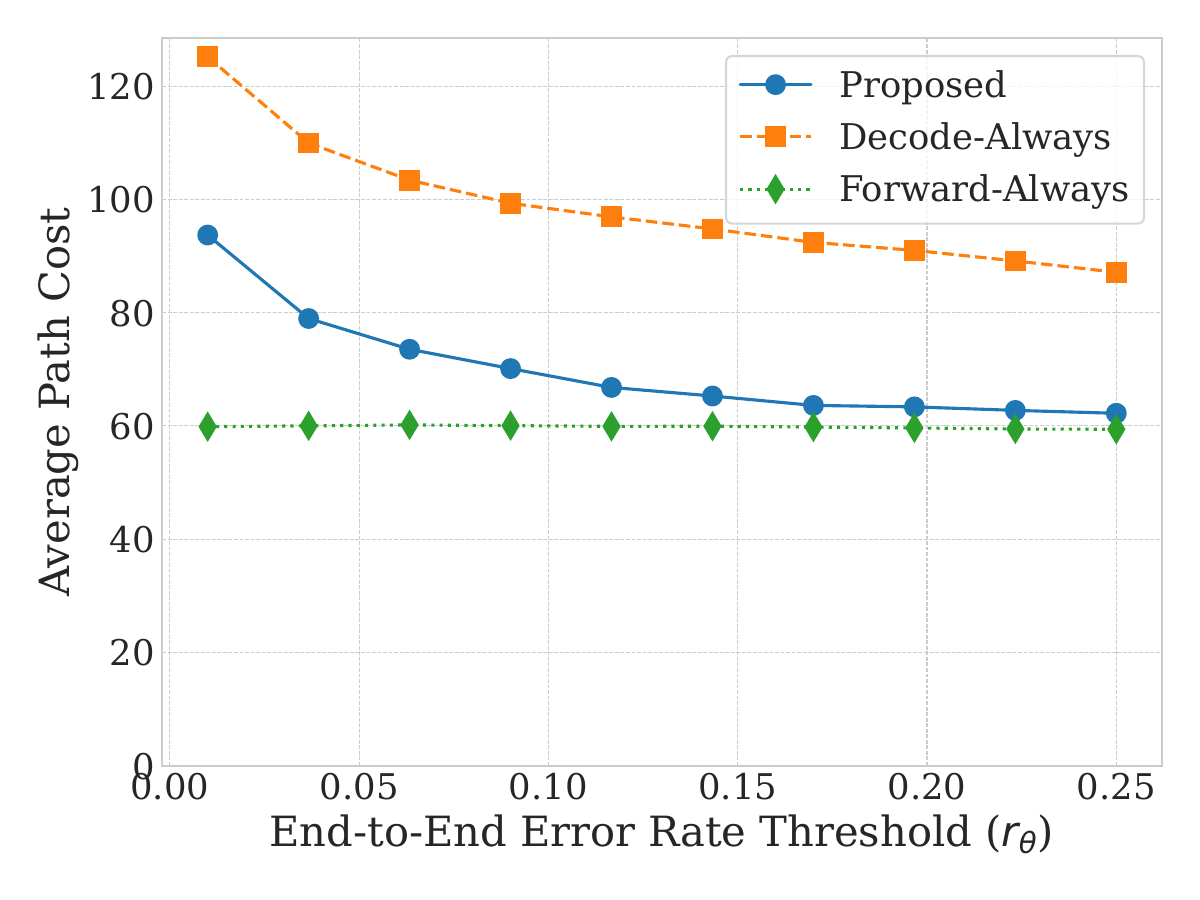}
    }
    \hfill
    \subfloat[Adaptive Scheme Assignment: Avg Cost\label{fig:exp2_cost}]{%
        \includegraphics[width=\subfigwidththree]{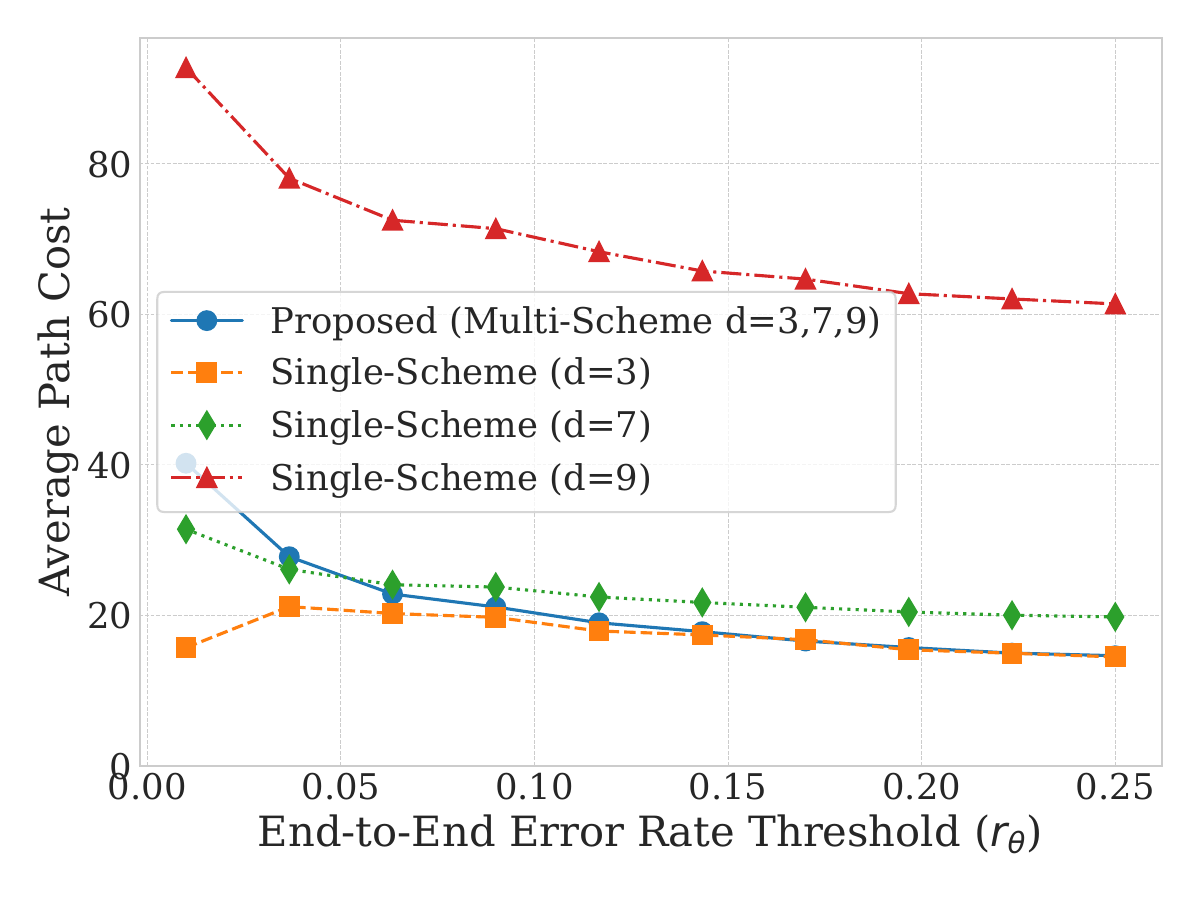}
    }
    \hfill
    \subfloat[Multi-Flow: Normalized Congestion\label{fig:exp3_congestion_unit}]{%
        \includegraphics[width=\subfigwidththree]{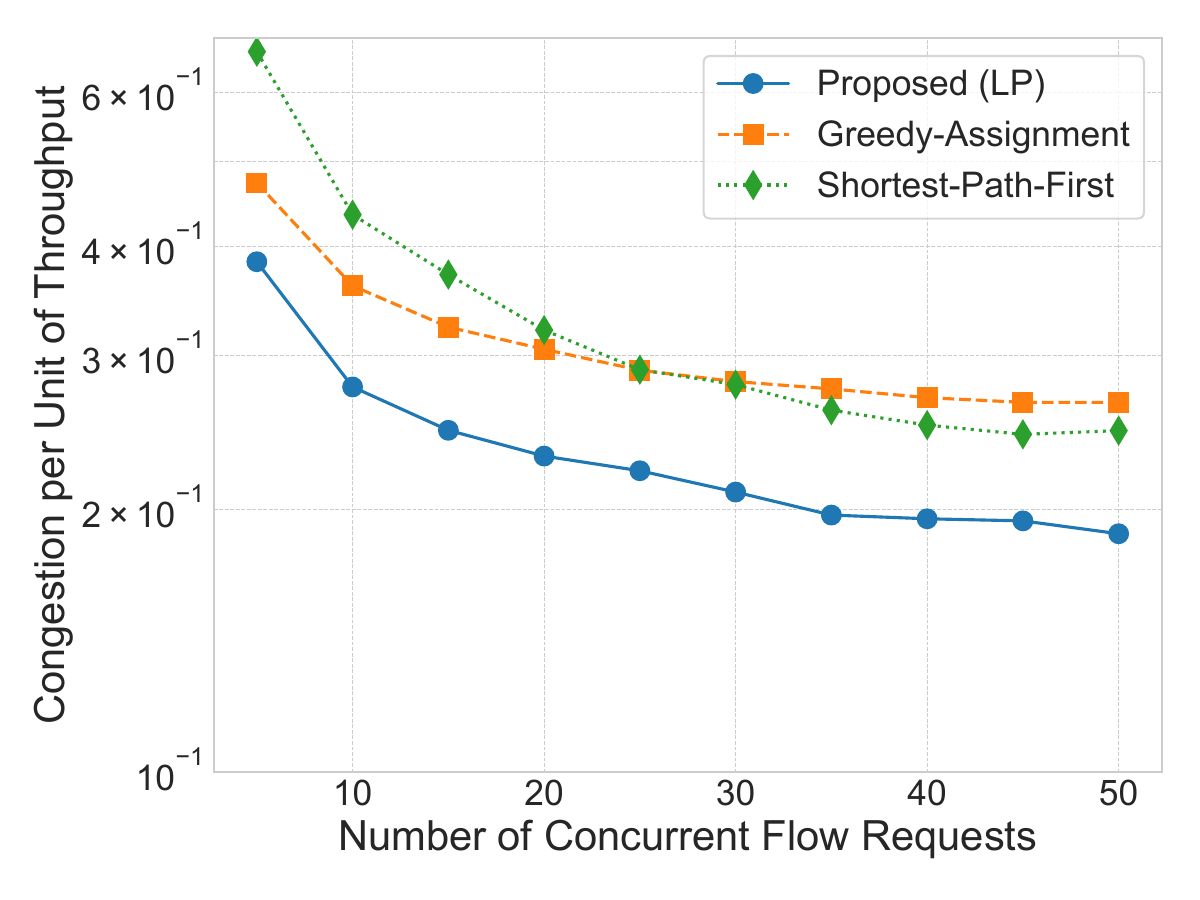}
    }


    \makebox[\textwidth][c]{%
        \subfloat[Running Time vs. Number of Nodes\label{fig:exp4_nodes}]{%
            \includegraphics[width=\subfigwidththree]{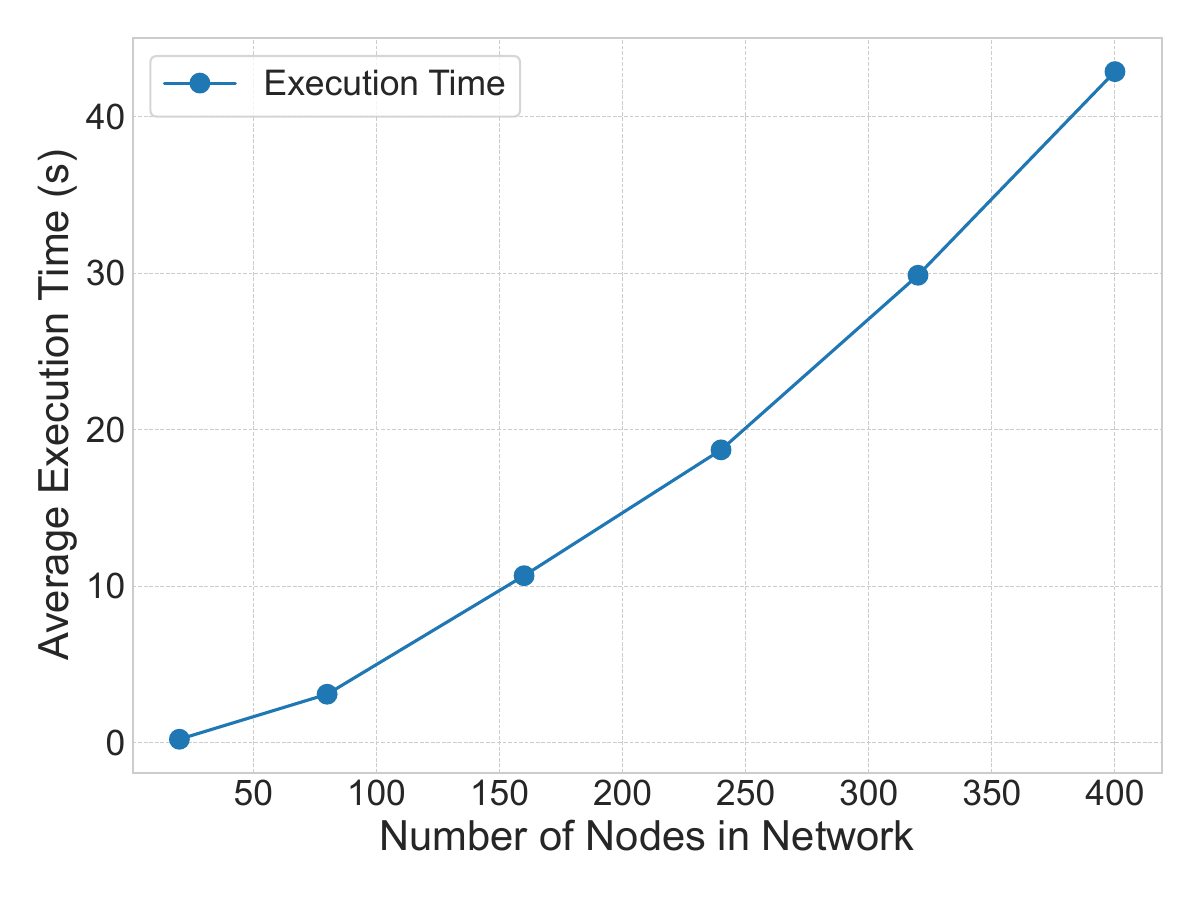}
        }
        \hspace{0.04\textwidth}
        \subfloat[Running Time vs. Number of Schemes\label{fig:exp4_schemes}]{%
            \includegraphics[width=\subfigwidththree]{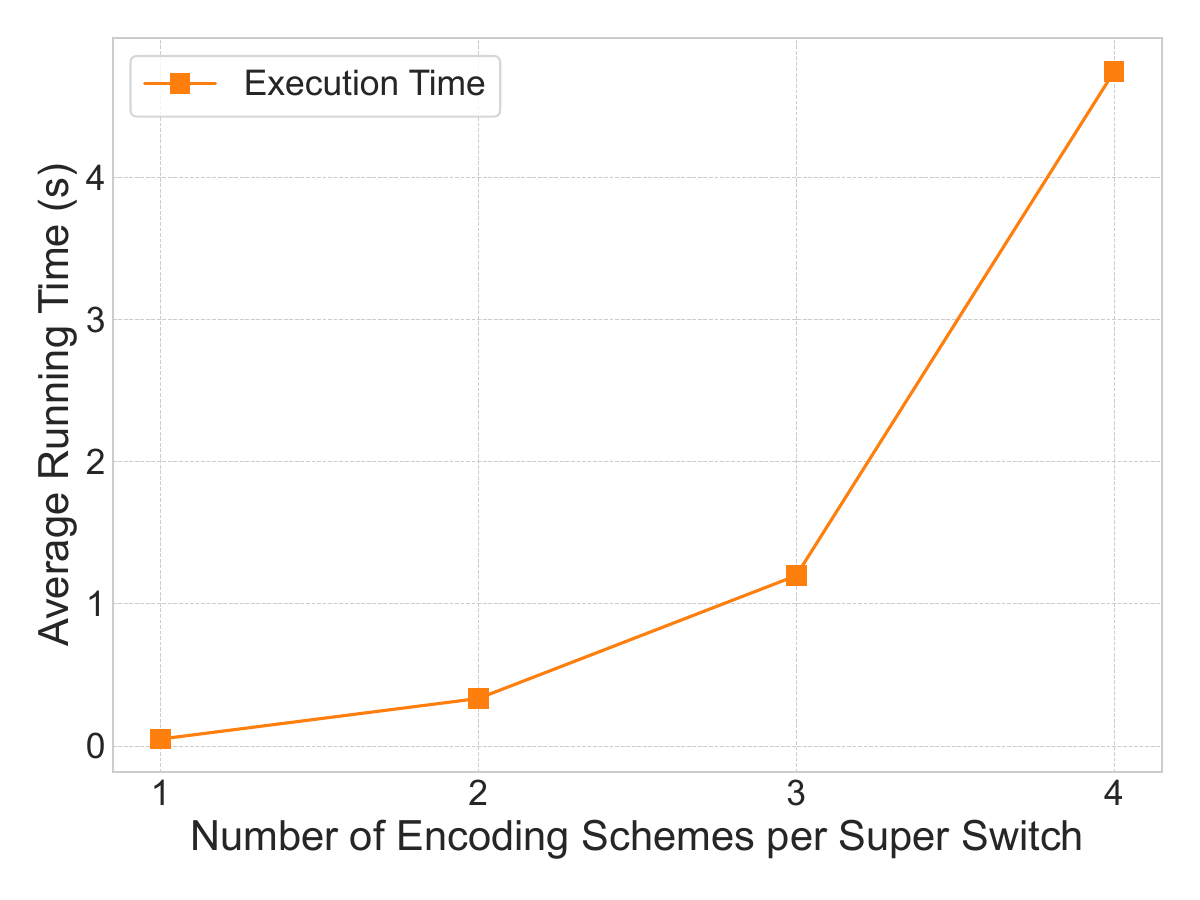}
        }
    }

    \caption{Performance evaluation: (a) and (d) validate the core algorithm; (b) and (e) evaluate adaptive scheme assignment; (c) and (f) report multi-flow performance; and (g) and (h) report running time.}
    \label{fig:all_experiments}
\end{figure*}

\subsection{Simulation Setup}

\subsubsection{Simulator}
%
%
We implement the evaluation in Python using \texttt{NetworkX}~\cite{hagberg2008networkx} for graph operations and \texttt{PuLP} with CBC solver for the ILP~\cite{mitchell2011pulp,forrest2005cbc}. It targets the routing abstraction rather than hardware and protocol events, which are the focus of \texttt{NetSquid}, \texttt{SeQUeNCe}, and \texttt{QuISP}.

\subsubsection{Network Topology}
%
%
We use the Barabási-Albert (BA) model to generate reproducible and scalable topologies with heterogeneous connectivity for controlled evaluation.
These topologies are not replicas of a specific deployment, and the results may vary for different topologies.
Unless otherwise specified, the four experiments use networks with 100, 100, 50, and 500 nodes, respectively.
The average degree is set to 2, and a fraction $p_{\text{super}}=0.4$ of nodes are randomly designated as QEC-capable super switches.

\subsubsection{Physical Parameters}
We model a metro-scale quantum network. For each channel $e$, its length $L_e$ is drawn from an exponential distribution with mean 2 km. The physical error rate $p_e$ is computed from photon loss under fiber attenuation $\alpha=0.15$ dB/km~\cite{agrawal2012fiber}:
\begin{equation}
    p_e(L_e) = \frac{3}{4} \left( 1 - 10^{-\frac{\alpha L_e}{10}} \right).
\end{equation}
The propagation time is calculated as $t_e=L_e/v_{\text{fiber}}$, where $v_{\text{fiber}}\approx 3.3\times 10^5$ m/s.

\subsubsection{Protection Schemes}
We use surface-code schemes with $d_{\text{code}}=3,5,7,9$; Table~\ref{tab:schemes} lists their qubit counts, costs, and decoherence times. We estimate $\bar f_{\sigma,j}$ via Monte Carlo decoder simulations with minimum-weight perfect matching (MWPM). For each $(\sigma,j)$, we sample and decode weight-$j$ Pauli patterns, marking a failure if the residual error induces a nontrivial logical operation.

\subsubsection{Compared Algorithms}
%
%
Decode-Always and Forward-Always are recovery-placement ablations.
Greedy-Assignment is an allocation-stage ablation, whereas SPF is an independent multi-flow baseline:
\begin{itemize}
    \item \textbf{Decode-Always}: A QEC-aware baseline that performs decoding at every super switch. It uses a simplified label-correcting algorithm that only explores paths with decoding at all available super switches.
%
%
    \item \textbf{Forward-Always}: Performs no intermediate recovery and treats intermediate super switches as forwarding-only nodes.
    \item \textbf{Greedy-Assignment}: Selects the lowest-cost strategy for each flow from the same candidate pool, without global congestion optimization.
    \item \textbf{Shortest-Path-First (SPF)}: An independent, cost-aware classical baseline. It computes each flow's minimum-cost physical path under the same link-cost metric, without our QEC-aware candidate generation or global congestion optimization, and then applies the same feasibility check.
\end{itemize}

%
%
Together, these baselines isolate the effects of recovery placement, per-flow path selection, and global allocation.
Classical MCPP solvers and the QEC approaches to be introduced in Section~\ref{sec:related-work} lack the same segment logical-error and inter-recovery lifetime constraints.
Comparing them requires surrogate metrics and recovery policies, which we leave to future cross-model evaluation.

\subsubsection{Performance Metrics}
We report five primary metrics: path acceptance ratio, average accepted-path cost, maximum channel congestion, congestion per unit throughput, and multi-scheme running time. The four experiments~(Fig.~\ref{fig:all_experiments}) average 500, 500, 100, and 500 independent runs, respectively.
%
%
The multi-flow experiment uses $M=3$, allowing up to three candidates per flow. A larger \(M\) increases both candidate diversity and computation.

\subsection{Core Algorithm Validation}

This experiment validates the fixed-scheme instantiation of the proposed framework in a single-flow setting using the $d_{\text{code}}=9$ surface code. 
Fig.~\ref{fig:exp1_acceptance} shows that the framework achieves an acceptance ratio comparable to Decode-Always and clearly higher than Forward-Always. 
Under strict error-rate thresholds, Proposed and Decode-Always accept about $85\%$ of requests, whereas Forward-Always accepts only about $43\%$. 
This confirms the necessity of intermediate QEC operations.

The cost advantage is shown in Fig.~\ref{fig:exp1_cost}. Compared with Decode-Always, the proposed framework reduces the average path cost by approximately $25$--$30\%$, e.g., from about $125$ to $94$ under the strictest threshold and from about $87$--$89$ to $62$ when $r_\theta=0.25$. This is because it performs decoding only when necessary, achieving a better cost-feasibility trade-off.

\subsection{Advantage of Adaptive Scheme Assignment}

This experiment evaluates the benefit of adaptive scheme assignment in a single-flow, multi-scheme setting. We compare the proposed multi-scheme algorithm with fixed single-scheme strategies using $d_{\text{code}}=3$, $7$, or $9$.

As shown in Fig.~\ref{fig:exp2_acceptance}, the weak $d_{\text{code}}=3$ scheme is inexpensive but unreliable under strict thresholds, accepting only about $8\%$ of requests, while Proposed and the robust $d_{\text{code}}=9$ scheme accept about $85\%$. Fig.~\ref{fig:exp2_cost} further shows that Proposed achieves this high feasibility at much lower cost: its average cost decreases from about $40$ to $15$ as $r_\theta$ increases, whereas the fixed $d_{\text{code}}=9$ strategy remains around $93$ to $62$. Thus, adaptive scheme assignment reduces cost by roughly $55$--$75\%$ compared with always using the strongest scheme, while preserving a similar acceptance ratio.

\subsection{Network-Level Performance (Multi-Flow)}

%
%
We next evaluate the network-level performance of the proposed framework in a multi-flow setting.
Fig.~\ref{fig:exp3_congestion} shows that raw maximum congestion rises with the number of concurrent flows for all methods. Since they accept different traffic volumes, Fig.~\ref{fig:exp3_congestion_unit} normalizes congestion by throughput.

At $50$ concurrent flows, Proposed reaches about $0.18$--$0.19$, versus about $0.26$ for Greedy-Assignment and $0.24$--$0.25$ for Shortest-Path-First.
These values represent reductions of approximately $28$--$31\%$ and $23$--$25\%$, respectively.
The Greedy-Assignment gap quantifies the benefit of global allocation.

\subsection{Running Time}

Finally, we evaluate the computational scalability of the proposed framework. As shown in Fig.~\ref{fig:exp4_nodes} and Fig.~\ref{fig:exp4_schemes}, the running time grows with both the network size $|V|$ and the number of available schemes $b$. It remains below $1$ second for networks with about $40$ nodes and grows to about $43$ seconds for $400$ nodes. The dependence on $b$ is sharper: the running time increases from about $0.05$ seconds for $b=1$ to about $1.2$ seconds for $b=3$ and about $4.7$ seconds for $b=4$, consistent with the $O(b^2)$ growth of auxiliary channels in the scheme-expanded graph.

\section{Related Work}
\label{sec:related-work}

%
%
Classical MCPP handles fixed additive edge weights using multi-label Bellman--Ford or Dijkstra variants~\cite{wang1995bandwidth,martins1984multicriteria}. 
Our formulation retains this search structure. However, it uses segment states and QEC-aware label transitions rather than fixed physical-channel logical-error weights, to adapt to quantum-network constraints~\cite{zhang2024quantumstack}.

Quantum routing either distributes entangled pairs or sends encoded qubits directly. 
The first paradigm creates end-to-end entanglement through generation, swapping, and purification. 
Prior work studies redundancy, fidelity-aware purification, multi-flow distribution, and opportunistic routing~\cite{zhao2021redundant,zhao2022e2e,chakraborty2020entanglement,farahbakhsh2022opportunistic}. Other studies address end-to-end fidelity, delay, distillation, latency, memory-limited repeaters, and purification scheduling~\cite{kumar2025routing,chehimi2025delay,liu2025statistical,vanmilligen2025entanglement,chen_optimum_2024}. 
Surveys further highlight quantum-specific constraints such as swapping probability, fidelity, optical loss, and short entangled-state lifetimes~\cite{abane2024entanglement}. 
These studies emphasize state quality and time but distribute entanglement rather than route encoded data qubits.

The second paradigm, which is the focus of this paper, sends quantum in\-for\-ma\-tion through physical channels and protects it with quantum error correction (QEC). 
Fowler et al. introduced surface-code-based quantum com\-mu\-ni\-ca\-tion~\cite{fowler2010surface}. 
Later work studied physical-qubit trans\-mis\-sion, including fidelity-based multicast~\cite{wang2019multicast}. 
More recently, SurfaceNet proposed a fault-tolerant ar\-chi\-tec\-ture that uses surface codes to preserve and transfer quantum messages~\cite{hu2024surfacenet}. 
Hu et al. showed that surface-code-aware routing improves com\-mu\-ni\-ca\-tion fidelity and network throughput~\cite{hu2024qnrouting}. 
Follow-up work optimizes surface-code structures for quantum-network error correction~\cite{hu2025disparity}. 
These works are closely related to our setting because they consider QEC-protected com\-mu\-ni\-ca\-tion and heterogeneous quantum-network resources. 
%
%
However, they use coarser fidelity or resource abstractions. They also do not jointly enforce segment logical error, recovery placement, scheme selection, and inter-recovery lifetime. 

Recent work also studies routing beyond entanglement distribution.
Circuit switching and packet switching have been considered for quantum secure direct communication networks~\cite{sun2025qsdcrouting}.
These studies broaden quantum-network design but do not address stabilizer-code-protected decode-and-re-encode routing.

In summary, prior work has laid important foundations for both paradigms of quantum networking. 
Nevertheless, the literature remains weighted toward entanglement distribution and repeater-based routing. 
Existing work on QEC-protected direct transmission has yet to capture how routing and recovery interact with accumulated physical noise, decoding placement, logical errors, and finite logical lifetime. 
This paper addresses this gap. 
We formulate a joint spatio-temporal path optimization problem for block-style stabilizer-code-protected direct transmission. 
Our framework uses a generic state space to jointly optimize path selection, decoding placement, and protection-scheme assignment under logical-error and decoherence-time constraints.

\section{Conclusion}
\label{sec:conclusion}

%
%
In this paper, we studied direct transmission of encoded qubits over noisy quantum networks. Unlike entanglement-based teleportation, direct transmission requires joint selection of paths, recovery locations, and protection schemes.
We formulated this spatio-temporal problem with non-additive segment logical error and lifetime constraints, and developed a framework for fixed and flexible schemes and multi-flow routing.
Simulations show a better cost-feasibility trade-off and lower normalized congestion.
The model assumes block-style node-local recovery.
Future work will include recovery faults and latency, realistic correlated noise, other QEC architectures, and multi-path routing. A broader evaluation will cover additional topology models, statistical error bars, sensitivity to size of multi-flow pool, and stronger multi-flow baselines.


\bibliographystyle{IEEEtran}
\bibliography{references}

\ifdefined\ICNPShepherdVersion
\appendices

\section{First-Order Approximation of Segment-Level Pauli Errors}
\label{ap:first-order-composite-error}

We justify the first-order approximation in our composite channel model. 
Consider a path segment $q$, and let $E(q)$ denote its set of physical channels. 
For each channel $e\in E(q)$, let $P_e$ denote the Pauli error that it introduces on a physical qubit, where
\[
P_e\in\{I,X,Y,Z\}.
\]
Let
\[
\Pr(P_e\neq I)=p_e.
\]
Equivalently, if the error is represented in the Pauli basis, then
\[
p_e=p_{e,\text{X}}+p_{e,\text{Y}}+p_{e,\text{Z}},
\]
where $p_{e,\text{X}}$, $p_{e,\text{Y}}$, and $p_{e,\text{Z}}$ are the probabilities of X, Y, and Z errors on channel $e$, respectively. We assume that errors introduced by different channels are independent.

The exact Pauli operator accumulated over segment $q$ is the product
\[
P_q=\prod_{e\in E(q)} P_e,
\]
where multiplication is taken modulo global phase. The exact net Pauli error probability of the segment is
\[
p_q^{\mathrm{net}}=\Pr(P_q\neq I).
\]
In contrast, the probability that at least one physical error occurs in the segment is
\[
p_q^{\mathrm{occ}}
=
\Pr\left(\exists e\in E(q): P_e\neq I\right)
=
1-\prod_{e\in E(q)}(1-p_e).
\]
Clearly, $p_q^{\mathrm{occ}}$ may overestimate $p_q^{\mathrm{net}}$, because multiple non-identity Pauli errors can compose to the identity.

\begin{proposition}[Pauli-cancellation error is second order]
\label{prop:pauli-cancellation-second-order}
Let
\[
S_q=\sum_{e\in E(q)}p_e.
\]
Then
\begin{equation*}
\label{eq:pauli-cancellation-bound}
\begin{aligned}
0 \le p_q^{\text{occ}}-p_q^{\text{net}}
&\le \Pr\!\left(\sum_{e\in E(q)} \mathbf{1}_{\{P_e\neq I\}} \ge 2\right) \\
&\le \sum_{\{e,e'\}\subseteq E(q)} p_e p_{e'} \\
&\le \frac{S_q^2}{2}.
\end{aligned}
\end{equation*}
Consequently,
\[
p_q^{\mathrm{net}}
=
p_q^{\mathrm{occ}}+O(S_q^2)
=
\sum_{e\in E(q)}p_e+O(S_q^2).
\]
Thus, the approximation error caused by Pauli cancellation is $O(S_q^2)$.
\end{proposition}

\begin{proof}
The event $\{P_q\neq I\}$ implies that at least one non-identity Pauli error has occurred, but the converse is not always true because multiple Pauli errors may compose to the identity. Therefore,
\[
p_q^{\mathrm{net}}\le p_q^{\mathrm{occ}},
\]
which gives
\[
p_q^{\mathrm{occ}}-p_q^{\mathrm{net}}\ge 0.
\]

Next, observe that $p_q^{\mathrm{occ}}$ and $p_q^{\mathrm{net}}$ can differ only when at least two non-identity Pauli errors occur in the segment. Indeed, if no physical error occurs, then $P_q=I$; if exactly one physical error occurs, then $P_q\in\{X,Y,Z\}$ and hence $P_q\neq I$. Therefore, cancellation to the identity is possible only when at least two physical errors occur. Let
\[
N_q=\sum_{e\in E(q)}\mathbf{1}_{\{P_e\neq I\}}
\]
be the number of non-identity physical errors in segment $q$. Then
\[
p_q^{\mathrm{occ}}-p_q^{\mathrm{net}}
\le
\Pr(N_q\ge 2).
\]

By the union bound over all unordered pairs of channels,
\[
\Pr(N_q\ge 2)
\le
\sum_{\{e,e'\}\subseteq E(q)}\Pr(P_e\neq I, P_{e'}\neq I).
\]
Under the independent channel-error assumption,
\[
\Pr(P_e\neq I, P_{e'}\neq I)=p_ep_{e'}.
\]
Thus,
\[
\Pr(N_q\ge 2)
\le
\sum_{\{e,e'\}\subseteq E(q)}p_ep_{e'}.
\]
Finally,
\[
\sum_{\{e,e'\}\subseteq E(q)}p_ep_{e'}
\le
\frac{1}{2}\left(\sum_{e\in E(q)}p_e\right)^2
=
\frac{S_q^2}{2}.
\]
Therefore,
\[
0\le p_q^{\mathrm{occ}}-p_q^{\mathrm{net}}
\le
\frac{S_q^2}{2}.
\]

It remains to relate these quantities to the first-order term. Since
\[
p_q^{\mathrm{occ}}
=
1-\prod_{e\in E(q)}(1-p_e),
\]
its inclusion-exclusion expansion gives
\[
p_q^{\mathrm{occ}}
=
\sum_{e\in E(q)}p_e+O(S_q^2).
\]
Combining this with
\[
p_q^{\mathrm{net}}=p_q^{\mathrm{occ}}+O(S_q^2)
\]
yields
\[
p_q^{\mathrm{net}}
=
\sum_{e\in E(q)}p_e+O(S_q^2).
\]
This proves that Pauli cancellation contributes an $O(S_q^2)$ approximation error in the low-error regime.
\end{proof}

\section{Exact Logical-Error Composition and Approximation Chain}
\label{ap:exact-logical-composition}
%
%
This appendix separates the two approximations underlying Eq.~\eqref{eq:r_new}. The complete approximation chain compares the exact-input logical-Pauli model
\[
p_q^{\mathrm{net}}
\xrightarrow{F_\sigma}
p_{q,\sigma}^{\log,\mathrm{net}}
\longrightarrow
r_{\mathrm{new}}^{\mathrm{exact}}
\]
with the routing model
\[
p_q^{\mathrm{comp}}
\xrightarrow{F_\sigma}
p_{q,\sigma}^{\log}
\longrightarrow
r_{\mathrm{new}}^{\mathrm{model}}.
\]
The first difference arises from using $p_q^{\mathrm{comp}}$ as a surrogate for the net segment-level Pauli error probability $p_q^{\mathrm{net}}$. The second replaces exact logical-Pauli composition with the scalar recurrence in Eq.~\eqref{eq:r_new}. We first propagate the upstream segment-model error and then combine it with the recurrence-composition bound.
The additive transformation of Eq.~\eqref{eq:r_new} applies to the routing-level scalar recurrence; the exact logical-Pauli composition below also depends on the error-type distributions through the cancellation term.

For scheme $\sigma$, define the logical-error mapping
\begin{equation}
F_\sigma(x)=\sum_{j=0}^{n_\sigma}\bar f_{\sigma,j}\binom{n_\sigma}{j}(1-x)^{n_\sigma-j}x^j.
\label{eq:F-sigma}
\end{equation}
The output of Eq.~\eqref{eq:simple-logical-error-rate-of-composite-channel} remains
\[
p_{q,\sigma}^{\log}=F_\sigma(p_q^{\mathrm{comp}}).
\]
Using the exact net Pauli error probability $p_q^{\mathrm{net}}$ and $S_q$ defined in Appendix~\ref{ap:first-order-composite-error}, introduce the accurate-input quantity
\[
p_{q,\sigma}^{\log,\mathrm{net}}=F_\sigma(p_q^{\mathrm{net}}).
\]
Let
\[
L_\sigma=\max_{x\in[0,1]}|F_\sigma'(x)|.
\]
The derivative of the Bernstein polynomial in Eq.~\eqref{eq:F-sigma} is
\[
F_\sigma'(x)=n_\sigma\sum_{j=0}^{n_\sigma-1}
(\bar f_{\sigma,j+1}-\bar f_{\sigma,j})
\binom{n_\sigma-1}{j}(1-x)^{n_\sigma-1-j}x^j.
\]
Because the Bernstein basis is nonnegative and sums to one on $[0,1]$,
\begin{equation}
L_\sigma\le n_\sigma\max_{0\le j<n_\sigma}
|\bar f_{\sigma,j+1}-\bar f_{\sigma,j}|.
\label{eq:L-sigma-bound}
\end{equation}
Proposition~\ref{prop:pauli-cancellation-second-order} and the mean-value theorem then give
\begin{equation}
|p_{q,\sigma}^{\log}-p_{q,\sigma}^{\log,\mathrm{net}}|
\le L_\sigma|p_q^{\mathrm{comp}}-p_q^{\mathrm{net}}|
\le \frac{L_\sigma S_q^2}{2}.
\label{eq:upstream-logical-bound}
\end{equation}

Let
\[
\mathcal{L}=\{I,\bar{X},\bar{Y},\bar{Z}\}
\]
denote the logical Pauli operators modulo global phase. Suppose that before traversing a path segment $q$, the residual logical error on the logical qubit is described by a distribution
\[
\mu^{\mathrm{prev}} = \bigl(\mu_I^{\mathrm{prev}}, \mu_{\bar{X}}^{\mathrm{prev}}, \mu_{\bar{Y}}^{\mathrm{prev}}, \mu_{\bar{Z}}^{\mathrm{prev}}\bigr),
\]
where
\[
\mu_I^{\mathrm{prev}}+\mu_{\bar{X}}^{\mathrm{prev}}+\mu_{\bar{Y}}^{\mathrm{prev}}+\mu_{\bar{Z}}^{\mathrm{prev}}=1.
\]
Let
\[
r_{\mathrm{prev}} = \mu_{\bar{X}}^{\mathrm{prev}}+\mu_{\bar{Y}}^{\mathrm{prev}}+\mu_{\bar{Z}}^{\mathrm{prev}} = 1-\mu_I^{\mathrm{prev}}
\]
be the pre-segment logical failure probability.

For a fixed segment $q$ and decoding scheme $\sigma$, let
\[
\lambda^{(q,\sigma)} = \bigl(\lambda_I^{(q,\sigma)}, \lambda_{\bar{X}}^{(q,\sigma)}, \lambda_{\bar{Y}}^{(q,\sigma)}, \lambda_{\bar{Z}}^{(q,\sigma)}\bigr)
\]
be the logical Pauli distribution induced by the segment and subsequent decoding/re-encoding operation when the exact net segment input is used, where
\[
\lambda_I^{(q,\sigma)}+\lambda_{\bar{X}}^{(q,\sigma)}+\lambda_{\bar{Y}}^{(q,\sigma)}+\lambda_{\bar{Z}}^{(q,\sigma)}=1,
\]
and
\[
p_{q,\sigma}^{\log,\mathrm{net}} = \lambda_{\bar{X}}^{(q,\sigma)}+\lambda_{\bar{Y}}^{(q,\sigma)}+\lambda_{\bar{Z}}^{(q,\sigma)} = 1-\lambda_I^{(q,\sigma)}.
\]

\begin{proposition}[Exact logical-error composition]
\label{prop:exact-composition}
Under the logical-Pauli abstraction above, the exact post-segment logical Pauli distribution is given by the Pauli-group convolution
\[
\mu_L^{\mathrm{new}} = \sum_{\substack{A,B\in\mathcal{L}:\\ AB=L}} \mu_A^{\mathrm{prev}}\lambda_B^{(q,\sigma)}, \qquad \forall L\in\mathcal{L},
\]
where the multiplication is taken modulo global phase. In particular, the exact post-segment logical failure probability satisfies
\[
r_{\mathrm{new}}^{\mathrm{exact}} = r_{\mathrm{prev}} + (1-r_{\mathrm{prev}})p_{q,\sigma}^{\log,\mathrm{net}} - \sum_{L\in\{\bar{X},\bar{Y},\bar{Z}\}} \mu_L^{\mathrm{prev}}\lambda_L^{(q,\sigma)}.
\]
\end{proposition}

\begin{proof}
The exact post-segment logical Pauli operator is obtained by composing the pre-existing logical Pauli error with the residual logical Pauli error induced by the new segment. Therefore, the post-segment distribution is the convolution of $\mu^{\mathrm{prev}}$ and $\lambda^{(q,\sigma)}$ over the logical Pauli group.

To derive the exact logical failure probability, it suffices to compute the probability that the post-segment logical operator equals the identity. Since for Pauli operators modulo global phase we have $AB=I$ if and only if $A=B$,
\[
\begin{aligned}
\mu_I^{\mathrm{new}}
= {} & \mu_I^{\mathrm{prev}}\lambda_I^{(q,\sigma)}
 + \mu_{\bar{X}}^{\mathrm{prev}}\lambda_{\bar{X}}^{(q,\sigma)} \\
& + \mu_{\bar{Y}}^{\mathrm{prev}}\lambda_{\bar{Y}}^{(q,\sigma)}
 + \mu_{\bar{Z}}^{\mathrm{prev}}\lambda_{\bar{Z}}^{(q,\sigma)}.
\end{aligned}
\]
Hence
\begin{align*}
r_{\mathrm{new}}^{\mathrm{exact}}
&= 1-\mu_I^{\mathrm{new}} \\
&= 1-\Bigl(
\mu_I^{\mathrm{prev}}\lambda_I^{(q,\sigma)}
 + \mu_{\bar{X}}^{\mathrm{prev}}\lambda_{\bar{X}}^{(q,\sigma)} \\
&\qquad
 + \mu_{\bar{Y}}^{\mathrm{prev}}\lambda_{\bar{Y}}^{(q,\sigma)}
 + \mu_{\bar{Z}}^{\mathrm{prev}}\lambda_{\bar{Z}}^{(q,\sigma)}
\Bigr).
\end{align*}
Using
\[
\mu_I^{\mathrm{prev}}=1-r_{\mathrm{prev}}, \qquad \lambda_I^{(q,\sigma)}=1-p_{q,\sigma}^{\log,\mathrm{net}},
\]
we obtain
\[
\begin{aligned}
r_{\mathrm{new}}^{\mathrm{exact}}
= {} & r_{\mathrm{prev}} + (1-r_{\mathrm{prev}})p_{q,\sigma}^{\log,\mathrm{net}} \\
& - \sum_{L\in\{\bar{X},\bar{Y},\bar{Z}\}}
\mu_L^{\mathrm{prev}}\lambda_L^{(q,\sigma)}.
\end{aligned}
\]
This proves the claim.
\end{proof}

\begin{proposition}[Recurrence-composition bound]
\label{prop:second-order}
Let
\[
r_{\mathrm{new}}^{\mathrm{approx}} = r_{\mathrm{prev}} + (1-r_{\mathrm{prev}})p_{q,\sigma}^{\log,\mathrm{net}}
\]
denote the scalar recurrence with the accurate segment logical-error input. Then
\[
0 \le r_{\mathrm{new}}^{\mathrm{approx}} - r_{\mathrm{new}}^{\mathrm{exact}} \le r_{\mathrm{prev}}\,p_{q,\sigma}^{\log,\mathrm{net}}.
\]
Consequently, if
\[
r_{\mathrm{prev}}\le \varepsilon \qquad\text{and}\qquad p_{q,\sigma}^{\log,\mathrm{net}}\le \varepsilon,
\]
then
\[
\bigl|r_{\mathrm{new}}^{\mathrm{approx}} - r_{\mathrm{new}}^{\mathrm{exact}}\bigr| \le \varepsilon^2,
\]
i.e., the approximation error is second-order in the low-error regime.
\end{proposition}

\begin{proof}
By Proposition~\ref{prop:exact-composition},
\[
\begin{aligned}
r_{\mathrm{new}}^{\mathrm{approx}} - r_{\mathrm{new}}^{\mathrm{exact}}
= {} & \sum_{L\in\{\bar{X},\bar{Y},\bar{Z}\}}
\mu_L^{\mathrm{prev}}\lambda_L^{(q,\sigma)}.
\end{aligned}
\]
Since every term in the above sum is nonnegative, we have
\[
r_{\mathrm{new}}^{\mathrm{approx}} - r_{\mathrm{new}}^{\mathrm{exact}} \ge 0.
\]
Moreover,
\begin{align*}
r_{\mathrm{new}}^{\mathrm{approx}} - r_{\mathrm{new}}^{\mathrm{exact}}
&= \sum_{L\in\{\bar{X},\bar{Y},\bar{Z}\}} \mu_L^{\mathrm{prev}}\lambda_L^{(q,\sigma)} \\
&\le \Bigl(
\mu_{\bar{X}}^{\mathrm{prev}}
+\mu_{\bar{Y}}^{\mathrm{prev}}
+\mu_{\bar{Z}}^{\mathrm{prev}}
\Bigr) \\
&\quad \cdot \Bigl(
\lambda_{\bar{X}}^{(q,\sigma)}
+\lambda_{\bar{Y}}^{(q,\sigma)}
+\lambda_{\bar{Z}}^{(q,\sigma)}
\Bigr) \\
&= r_{\mathrm{prev}}\,p_{q,\sigma}^{\log,\mathrm{net}}.
\end{align*}
Therefore,
\[
0 \le r_{\mathrm{new}}^{\mathrm{approx}} - r_{\mathrm{new}}^{\mathrm{exact}} \le r_{\mathrm{prev}}\,p_{q,\sigma}^{\log,\mathrm{net}}.
\]
If both $r_{\mathrm{prev}}$ and $p_{q,\sigma}^{\log,\mathrm{net}}$ are at most $\varepsilon$, then the right-hand side is bounded by $\varepsilon^2$. This proves the claimed $O(\varepsilon^2)$ error bound.
\end{proof}

%
%
The update actually computed by Eqs.~\eqref{eq:p_comp}--\eqref{eq:r_new} is
\[
r_{\mathrm{new}}^{\mathrm{model}}
=r_{\mathrm{prev}}+(1-r_{\mathrm{prev}})p_{q,\sigma}^{\log}.
\]
Combining Eq.~\eqref{eq:upstream-logical-bound} with Proposition~\ref{prop:second-order} yields
\begin{equation}
\begin{aligned}
|r_{\mathrm{new}}^{\mathrm{model}}-r_{\mathrm{new}}^{\mathrm{exact}}|
&\le |r_{\mathrm{new}}^{\mathrm{model}}-r_{\mathrm{new}}^{\mathrm{approx}}|
 +|r_{\mathrm{new}}^{\mathrm{approx}}-r_{\mathrm{new}}^{\mathrm{exact}}|\\
&\le (1-r_{\mathrm{prev}})\frac{L_\sigma S_q^2}{2}
 +r_{\mathrm{prev}}p_{q,\sigma}^{\log,\mathrm{net}}.
\end{aligned}
\label{eq:complete-chain-bound}
\end{equation}
Under the independent-channel and depolarizing assumptions used above, the complete one-step error is $O(\varepsilon^2)$ when $S_q=O(\varepsilon)$, $r_{\mathrm{prev}}=O(\varepsilon)$, $p_{q,\sigma}^{\log,\mathrm{net}}=O(\varepsilon)$, and $L_\sigma=O(1)$. If these conditions hold over $m$ recovery segments, summing the one-step bounds gives the conservative accumulated bound $O(m\varepsilon^2)$. Therefore, under these low-error conditions, the first-order approximation chain used by the routing model has $O(\varepsilon^2)$ one-step error relative to exact logical-Pauli composition.

\begin{corollary}[Symmetric logical depolarizing special case]
\label{cor:symmetric-case}
If the pre-segment residual logical error and the segment-induced logical error are both symmetric over the three non-identity logical Pauli operators, i.e.,
\[
\mu_{\bar{X}}^{\mathrm{prev}}=\mu_{\bar{Y}}^{\mathrm{prev}}=\mu_{\bar{Z}}^{\mathrm{prev}}=\frac{r_{\mathrm{prev}}}{3},
\]
and
\[
\lambda_{\bar{X}}^{(q,\sigma)}=\lambda_{\bar{Y}}^{(q,\sigma)}=\lambda_{\bar{Z}}^{(q,\sigma)}=\frac{p_{q,\sigma}^{\log,\mathrm{net}}}{3},
\]
then
\[
r_{\mathrm{new}}^{\mathrm{exact}} = r_{\mathrm{prev}} + (1-r_{\mathrm{prev}})p_{q,\sigma}^{\log,\mathrm{net}} - \frac{r_{\mathrm{prev}}p_{q,\sigma}^{\log,\mathrm{net}}}{3}.
\]
Equivalently,
\[
r_{\mathrm{new}}^{\mathrm{exact}} = r_{\mathrm{prev}} + p_{q,\sigma}^{\log,\mathrm{net}} - \frac{4}{3}r_{\mathrm{prev}}p_{q,\sigma}^{\log,\mathrm{net}}.
\]
\end{corollary}

\begin{proof}
Under the symmetry assumptions,
\[
\sum_{L\in\{\bar{X},\bar{Y},\bar{Z}\}} \mu_L^{\mathrm{prev}}\lambda_L^{(q,\sigma)} = 3\cdot \frac{r_{\mathrm{prev}}}{3}\cdot \frac{p_{q,\sigma}^{\log,\mathrm{net}}}{3} = \frac{r_{\mathrm{prev}}p_{q,\sigma}^{\log,\mathrm{net}}}{3}.
\]
Substituting this identity into the exact composition formula in Proposition~\ref{prop:exact-composition} yields the desired expression.
\end{proof}

\section{Proof of Theorem~\ref{th:NP-hard}}
\label{ap:proof-of-theorem-1}

In this appendix, we prove Theorem~\ref{th:NP-hard}. We first introduce the decision version of the problem and show that it is NP-complete. This immediately implies that the corresponding optimization problem is NP-hard.

\medskip
\textbf{Min-Cost Feasible Strategy Decision Problem (MCFSD).}

\textit{Input:} A network instance under our model, a source-destination pair $(s,d)$, a cost budget $B$, an end-to-end logical error threshold $r_\theta$, and the code/time parameters.

\textit{Question:} Does there exist a strategy $\pi=(P,D,\Sigma)$ such that
\begin{enumerate}
    \item $P$ is an $s$-$d$ simple path;
    \item the total strategy cost satisfies $C(\pi)\le B$;
    \item the final logical error rate satisfies $r(\pi)\le r_\theta$; and
    \item each segment delimited by two consecutive decoding operations satisfies the logical decoherence constraint?
\end{enumerate}

\begin{proof}[Proof of Theorem~\ref{th:NP-hard}]
%
%
We prove NP-hardness by exploiting the decision version \textsc{MCFSD}. Specifically, we show that \textsc{MCFSD} is NP-complete by a polynomial-time reduction from the Shortest Weight-Constrained Path Problem (SWCPP). Since a polynomial-time algorithm for the optimization problem would also solve \textsc{MCFSD}, the optimization problem is NP-hard.

\textbf{Step 1: \textsc{MCFSD} is in NP:}
A certificate is a strategy $\pi=(P,D,\Sigma)$. We can verify in polynomial time that:
\begin{itemize}
    \item $P$ is a simple $s$-$d$ path;
    \item $D$ is a valid set of decoding locations on $P$;
    \item $\Sigma$ assigns an available code scheme to each node in $D$;
    \item the total cost $C(\pi)$ is at most $B$;
    \item every segment satisfies the decoherence-time constraint; and
    \item the final logical error rate is at most $r_\theta$.
\end{itemize}
Indeed, all these checks require only a polynomial-time scan of the path and its segments, together with polynomial-time evaluation of the cost function, the segment-time constraints, and the logical-error update rule. Here we treat the code parameters, lookup tables, and the coefficients used by the logical-error model as either fixed constants of the instance or data that can be accessed in polynomial time. Therefore, \textsc{MCFSD} belongs to NP.

\textbf{Step 2: NP-hardness:}
We reduce from SWCPP, a known NP-complete problem~\cite{garey_computers_2009}. The description of the SWCPP is as follows.

\medskip
\noindent
\textbf{Shortest Weight-Constrained Path Problem~(SWCPP).}

\textit{Input:} A graph $G=(V,E)$, a source $s\in V$, a destination $d\in V$, an integer channel cost $c_e\ge 0$ for each $e\in E$, an integer channel weight $w_e\ge 0$ for each $e\in E$, and two integer budgets $B,W\ge 0$.

\textit{Question:} Does there exist an $s$-$d$ simple path $P$ such that
\[
\sum_{e\in P} c_e \le B \qquad\text{and}\qquad \sum_{e\in P} w_e \le W?
\]

\textbf{Reduction construction:} Given any SWCPP instance, we build an instance of \textsc{MCFSD} as follows.

\begin{enumerate}
    \item \textbf{Underlying graph.} Use the same graph $G=(V,E)$ with the same source $s$ and destination $d$.

    \item \textbf{Node types.} Let only $s$ and $d$ be super switches. Every node in $V\setminus\{s,d\}$ is set to be an ordinary switch. Thus, no intermediate node can perform decoding.

    \item \textbf{Code schemes.} Equip both $s$ and $d$ with exactly one code scheme $\sigma$. Set its decoding cost to
    \[
    c_\sigma = 0.
    \]

    \item \textbf{Channel costs.} For every channel $e\in E$, set its traversal cost in our model to be exactly the SWCPP cost:
    \[
    c'_e = c_e.
    \]

    \item \textbf{Propagation times.} For every channel $e\in E$, set its propagation time to be exactly the SWCPP weight:
    \[
    t_e = w_e.
    \]

    \item \textbf{Physical error rates.} For every channel $e\in E$, set
    \[
    p_e = 0.
    \]
    Hence, every composite physical error rate is zero, and therefore the final logical error rate is zero as well.

    \item \textbf{Logical decoherence time.} Set the logical decoherence time of the unique code scheme $\sigma$ to
    \[
    T_\sigma^{\log} = W+1.
    \]
\end{enumerate}

This construction is clearly computable in polynomial time.

\textbf{Key property of the construction:} Due to all intermediate nodes are ordinary switches, decoding can only occur at $d$ (with the source $s$ serving as the starting point of the transmission). 
Therefore, for any chosen $s$-$d$ path $P$, the entire path forms exactly one segment.

Thus, the segment-time constraint in our model becomes
\[
\sum_{e\in P} t_e < T_\sigma^{\log}.
\]
By construction, $t_e=w_e$ and $T_\sigma^{\log}=W+1$, so this is equivalent to
\[
\sum_{e\in P} w_e < W+1.
\]
Since all $w_e$ are integers, the above strict inequality holds if and only if
\[
\sum_{e\in P} w_e \le W.
\]

Moreover, since all $p_e=0$, the logical error rate along any path is zero. Hence, the end-to-end logical error constraint is automatically satisfied for any nonnegative threshold $r_\theta$ (for example, one may simply set $r_\theta=0$).

Finally, because the unique code scheme has zero decoding cost, the total strategy cost equals the path cost:
\[
C(\pi) = \sum_{e\in P} c'_e = \sum_{e\in P} c_e.
\]
Therefore,
\[
C(\pi)\le B \qquad\Longleftrightarrow\qquad \sum_{e\in P} c_e \le B.
\]

\textbf{Correctness of the reduction.} We prove both directions of the equivalence.

\medskip
\noindent
\textbf{(\(\Rightarrow\))} Suppose the SWCPP instance is a YES-instance. Then there exists an $s$-$d$ path $P$ such that
\[
\sum_{e\in P} c_e \le B \qquad\text{and}\qquad \sum_{e\in P} w_e \le W.
\]
Consider the corresponding strategy $\pi$ in the constructed \textsc{MCFSD} instance that follows the same path $P$. Since there is no intermediate super switch, the whole path is one segment. The segment-time constraint holds because
\[
\sum_{e\in P} t_e = \sum_{e\in P} w_e \le W < W+1 = T_\sigma^{\log}.
\]
The logical error constraint holds because all physical error rates are zero. The total strategy cost is
\[
C(\pi)=\sum_{e\in P} c_e \le B.
\]
Hence, the constructed \textsc{MCFSD} instance is also a YES-instance.

\medskip
\noindent
\textbf{(\(\Leftarrow\))} Conversely, suppose the constructed \textsc{MCFSD} instance is a YES-instance. Then there exists a feasible strategy $\pi$ with total cost at most $B$. Since only $s$ and $d$ are super switches, the strategy determines an $s$-$d$ path $P$ and the whole path is a single segment. Feasibility implies
\[
\sum_{e\in P} t_e < T_\sigma^{\log}=W+1.
\]
Using $t_e=w_e$ and the integrality of the weights, we obtain
\[
\sum_{e\in P} w_e \le W.
\]
Also,
\[
\sum_{e\in P} c_e = C(\pi)\le B.
\]
Therefore, $P$ is a feasible solution to the original SWCPP instance. So the SWCPP instance is a YES-instance.

\textbf{Conclusion.} We have shown that the SWCPP instance is a YES-instance if and only if the constructed \textsc{MCFSD} instance is a YES-instance. Thus, SWCPP polynomially reduces to \textsc{MCFSD}, so \textsc{MCFSD} is NP-hard. Combining this with the fact that \textsc{MCFSD} is in NP, we obtain that \textsc{MCFSD} is NP-complete.

Suppose, for contradiction, that a minimum-cost feasible strategy could be found in polynomial time. 
Then we could solve \textsc{MCFSD} in polynomial time by computing the minimum feasible cost and comparing it with the budget $B$. 
This contradicts the NP-completeness of \textsc{MCFSD}. Therefore, the optimization problem is NP-hard.
\end{proof}

\section{Proof of Theorem~\ref{th:mincost}}
\label{ap:proof-mincost}

\begin{proof}
Let $\hat{P}$ denote the path returned by Algorithm~\ref{alg:mincost} on $H=(V_H,E_H)$, and let $\hat r(\hat{P})$ and $\hat t(\hat{P})$ denote its discretized logical error rate and discretized segment time, respectively. Let $P^*$ be an optimal feasible path, and let $l^*$ be the label corresponding to $P^*$ at the destination.

First, we show that $c(\hat{P})\le c(P^*)$. During the label-correcting process, a label is removed only when it is dominated by another label with no larger cost, no larger logical error rate, and no larger elapsed segment time. Hence, pruning a dominated label cannot remove all minimum-cost feasible continuations. Therefore, the label corresponding to $P^*$ either remains in $\mathcal{L}_d$ or is dominated by another label in $\mathcal{L}_d$ with cost no larger than $c(P^*)$. Since Algorithm~\ref{alg:mincost} returns the minimum-cost label in $\mathcal{L}_d$, we have $c(\hat{P})\le c(P^*)$.

Next, we bound the error-rate violation introduced by error-rate discretization. Suppose $\Delta_r\le \epsilon_r r_\theta/|V_H|$. Since the logical error rate is rounded down with step size $\Delta_r$, each discretization step may underestimate the true error rate by at most $\Delta_r$. In the worst case, such rounding occurs at most $|V_H|$ times along a simple path in the auxiliary graph. Thus,
\[
r(\hat{P})\le \hat r(\hat{P})+\Delta_r|V_H|.
\]
Because Algorithm~\ref{alg:mincost} accepts the output label only when $\hat r(\hat{P})\le r_\theta$, we obtain
\[
r(\hat{P})\le r_\theta+\Delta_r|V_H|\le (1+\epsilon_r)r_\theta.
\]

Finally, we bound the violation of the logical decoherence-time constraint. Suppose $\Delta_t\le \epsilon_t T^{\text{log}}/|V_H|$. Since elapsed time is rounded down with step size $\Delta_t$, each discretization step may underestimate the true elapsed time by at most $\Delta_t$. Along any segment, the accumulated underestimation is therefore at most $\Delta_t|V_H|$. Hence, for any segment $q$ of $\hat{P}$,
\[
t^{\text{seg}}_q(\hat{P})\le \hat{t}^{\text{seg}}_q(\hat{P})+\Delta_t|V_H|.
\]
Because the algorithm accepts a segment only when $\hat{t}^{\text{seg}}_q(\hat{P})<T^{\text{log}}_q$, we have
\[
t^{\text{seg}}_q(\hat{P})<T^{\text{log}}_q+\Delta_t|V_H|\le (1+\epsilon_t)T^{\text{log}}_q.
\]
Therefore, $\hat{P}$ is an $(\epsilon_r,\epsilon_t)$-optimal feasible path.
\end{proof}

\section{Proof of the Multi-Flow Approximation Guarantee}
\label{ap:multi-flow-proof}

In this appendix, we prove the approximation guarantee of the randomized rounding policy used in Section~\ref{sec:multi-flows}. Recall that $\mathcal{K}$ is the set of flows, $\Pi_i$ is the fixed candidate strategy pool of flow $i$, and $a_i$ is the throughput demand of flow $i$. We assume that the throughput demands are normalized so that $0\le a_i\le 1$. Let $y_{\mathrm{OPT}}^{\Pi}$ denote the optimal value of the candidate-restricted integral path-allocation problem, and let $y^*$ denote the optimal value of its LP relaxation. Clearly, $y^*\le y_{\mathrm{OPT}}^{\Pi}$.

The randomized rounding policy independently selects one strategy $\pi\in\Pi_i$ for each flow $i$ with probability $x^*_{i,\pi}$, where $\{x^*_{i,\pi}\}$ is an optimal solution to the LP relaxation. For a physical channel $e\in E$, let $X_e$ denote the total rounded congestion on channel $e$. Then
\[
X_e=\sum_{i\in\mathcal{K}}\sum_{\pi\in\Pi_i:\,e\in E(\pi)} a_i X_{i,\pi},
\]
where $X_{i,\pi}$ is the indicator variable that flow $i$ selects strategy $\pi$. Since each flow is rounded independently, the random variables contributed by different flows are independent. Moreover, each contribution is bounded by $1$ due to the normalization $0\le a_i\le 1$.

The expected congestion on channel $e$ is
\[
\mathbb{E}[X_e]=\sum_{i\in\mathcal{K}}\sum_{\pi\in\Pi_i:\,e\in E(\pi)} a_i x^*_{i,\pi}\le y^*\le y_{\mathrm{OPT}}^{\Pi}.
\]
Let $\mu_e=\mathbb{E}[X_e]$. For any $\lambda>1$, by the Chernoff bound for independent bounded random variables,
\[
\Pr[X_e\ge \lambda y_{\mathrm{OPT}}^{\Pi}]
\le
\left(\frac{e\mu_e}{\lambda y_{\mathrm{OPT}}^{\Pi}}\right)^{\lambda y_{\mathrm{OPT}}^{\Pi}}.
\]
Since $\mu_e\le y_{\mathrm{OPT}}^{\Pi}$, we have
\[
\Pr[X_e\ge \lambda y_{\mathrm{OPT}}^{\Pi}]
\le
\left(\frac{e}{\lambda}\right)^{\lambda y_{\mathrm{OPT}}^{\Pi}}.
\]
Assuming $y_{\mathrm{OPT}}^{\Pi}\ge 1$, this further implies
\[
\Pr[X_e\ge \lambda y_{\mathrm{OPT}}^{\Pi}]
\le
\left(\frac{e}{\lambda}\right)^\lambda.
\]

Set $\lambda=\kappa\log n/\log\log n$ for a sufficiently large constant $\kappa$, where $n=|V|$. Then
\[
\left(\frac{e}{\lambda}\right)^\lambda
=
\exp\!\left(-\lambda\log(\lambda/e)\right)
\le
\frac{1}{n^3},
\]
for sufficiently large $n$ and a suitable constant $\kappa$. Therefore, for any fixed channel $e$,
\[
\Pr[X_e\ge \lambda y_{\mathrm{OPT}}^{\Pi}]\le \frac{1}{n^3}.
\]

Finally, by the union bound over all physical channels, and using $|E|\le n^2$, we obtain
\[
\Pr\left[\exists e\in E:\ X_e\ge \lambda y_{\mathrm{OPT}}^{\Pi}\right]
\le
\frac{|E|}{n^3}
\le
\frac{1}{n}.
\]
%
%
Thus, with probability at least $1-1/n$, every channel has congestion at most $\lambda y_{\mathrm{OPT}}^{\Pi}=O(\log n/\log\log n)\cdot y_{\mathrm{OPT}}^{\Pi}$. Hence, randomized rounding achieves the stated approximation ratio for the fixed candidate pools with high probability.

\fi
\end{document}